\documentclass{article}

\usepackage[nonatbib,preprint]{neurips_2021}

\usepackage[utf8]{inputenc} 
\usepackage[T1]{fontenc}    
\usepackage{hyperref}       
\usepackage{url}            
\usepackage{booktabs}       
\usepackage{amsfonts}       
\usepackage{nicefrac}       
\usepackage{microtype}      
\usepackage{xcolor}         

\usepackage{microtype}
\usepackage{graphicx}
\usepackage{subfigure}
\usepackage{booktabs} 
\usepackage{makecell}

\usepackage{xcolor}
\usepackage{amsmath}

\usepackage{amssymb}
\usepackage{amsthm}
\usepackage{bbm}

\usepackage{algorithm}
\usepackage[noend]{algorithmic}

\usepackage{enumitem} 

\newtheorem{thm}{Theorem}
\newtheorem{lemma}{Lemma}
\newtheorem{definition}{Definition}

\title{Practical Near Neighbor Search via Group Testing}

\author{ %
  Joshua Engels\thanks{Equal contribution.}\\
  Department of Computer Science\\
  Rice University\\
  Houston, Texas, USA \\
  \texttt{jae4@rice.edu} \\
  \And
  Benjamin Coleman\textsuperscript{*}\\
  Electrical and Computer Engineering\\
  Rice University\\
  Houston, Texas, USA \\
  \texttt{ben.coleman@rice.edu} \\
  \And 
  Anshumali Shrivastava\\
  Department of Computer Science\\
  Rice University\\
  Houston, Texas, USA \\
  \texttt{anshumali@rice.edu} \\
}

\begin{document}

\maketitle

\begin{abstract}
  We present a new algorithm for the approximate near neighbor problem that combines classical ideas from group testing with locality-sensitive hashing (LSH). We reduce the near neighbor search problem to a group testing problem by designating neighbors as ``positives,'' non-neighbors as ``negatives,'' and approximate membership queries as group tests. We instantiate this framework using distance-sensitive Bloom Filters to Identify Near-Neighbor Groups (FLINNG). We prove that FLINNG has sub-linear query time and show that our algorithm comes with a variety of practical advantages. For example, FLINNG can be constructed in a single pass through the data, consists entirely of efficient integer operations, and does not require any distance computations. We conduct large-scale experiments on high-dimensional search tasks such as genome search, URL similarity search, and embedding search over the massive YFCC100M dataset. In our comparison with leading algorithms such as HNSW and FAISS, we find that FLINNG can provide up to a 10x query speedup with substantially smaller indexing time and memory.
\end{abstract}

\section{Introduction}

Nearest neighbor search is a fundamental problem with many applications in machine learning systems. Informally, the task is as follows. Given a dataset $D = \{x_1, x_2, ... x_N\}$, we wish to build a data structure that can be queried with any point $q$ to obtain a small set of points $x_i \in D$ that have high similarity (low distance) to the query. This structure is called an \textit{index}. Near neighbor indices form the backbone of production models in recommendation systems, social networks, genomics, computer vision and many other application domains. 


\textbf{Applications:} In this paper, we focus on algorithms for approximate near neighbor search over high-dimensional large scale datasets. Such tasks frequently arise in genomics, web-scale data mining, machine learning, and other large-scale applications. Consider the Yahoo Flickr Creative Commons dataset (YFCC100M) which consists of 100 million media embeddings that are derived from the neuron activations for a convolutional neural network~\cite{thomee2016yfcc100m}. Each embedding is a 4096-dimensional vector. The dataset is about 1TB in size and presents a substantial challenge for even the most popular algorithms, which struggle with memory, index construction, and query time. Similar issues occur in genomics, where the task is to identify genome sequences with a high Jaccard similarity to the query. Modern genomic datasets can contain millions of reads with billions of possible $n$-gram sequences~\cite{ondov2016mash}. Many algorithms work well when there are a few hundred dimensions but are inappropriate for such applications. Our experiments demonstrate that for the datasets of interest in this paper, popular indices like HNSW and FAISS can take days to build, require gigabytes of RAM and have a suboptimal precision-recall-query time tradeoff.

Since our goal is to perform approximate search, dimensionality reduction is a reasonable strategy. However, dimensionality reduction is costly for ultra-high dimensional data. In genomics applications, $n$-gram sizes are typically very large $(n > 18)$. Thus, the one-hot encoding of each sequence can require billions of dimensions $(4^{18} \approx 68\text{B})$, making it intractable to learn an embedding model. For embedding applications such as YFCC100M or product search, a large embedding dimension can lead to performance improvements 
~\cite{medini2021solar}. Dimensionality reduction can incur a performance penalty, so we may wish to perform the near neighbor search over the original metric space.

Ideally, we would choose an algorithm that did not store data points in RAM, evaluate the distance function many times, employ iterative processes such as $k$-means, or construct complicated structures such as graphs, which are hard to parallelize and distribute. Recent algorithms such as FLASH~\cite{wang2018randomized} provide the ability to search based on aggregate LSH count statistics without computing distances, but these methods are heuristics that do not have theoretical guarantees. On the other hand, algorithms such as LSH, which have a well-established theoretical grounding, tend to perform poorly in practice because of their prohibitive hash table size and post filtering stage, which needs many ($N^\rho$ in theory) distance computations. In this paper, we present an algorithm having all the practical advantages of a system like FLASH while also being more accurate, theoretically sound, and provably sub-linear.

\subsection{Our Contribution}
In this paper, we address the computational challenges of high-dimensional similarity search by presenting an index with fast construction time, low memory requirement, and zero query-time distance computations.
Our approach is to transform a near neighbor search problem into a group testing problem by designing a test that outputs ``positive'' when a group of points contains a near neighbor to a query. That is, each test answers an approximate membership query over its group. Given a query, our algorithm produces a $B\times R$ array of group test results that can be efficiently decoded to identify the nearest neighbors. This is more efficient than statistical aggregation algorithms like FLASH because each test filters out entire groups of non-neighbors with a single test operation.

We develop a concrete example of such an algorithm by using Filters to Identify Near Neighbor Groups (FLINNG). We use a standard non-adaptive group testing design with distance-sensitive Bloom filters as tests. We prove that FLINNG solves the randomized nearest neighbor problem in $O\left(\log^2(\frac{1}{\delta})\log^3(N)N^{\frac{1}{2} + \gamma}\right)$ time, where $\gamma$ is a query-dependent parameter that characterizes query stability. We also implement FLINNG in C++ and conduct experiments on real-world high-dimensional datasets from genomics, embedding search, and URL analysis, where FLINNG achieves up to a 10x query speedup over existing indices with faster construction time and lower memory.


\section{Related Work}
The near neighbor problem has been the focus of more than four decades of intense research activity. The low dimensional problem is particularly well-understood, with space partitioning trees that can efficiently find the exact $k$-nearest neighbors. However, exact search in high dimensions is intractable due to the curse of dimensionality - the computational resources needed to solve the exact problem scale exponentially with dimensions. This has led to a diverse set of algorithms to solve the \textit{approximate near neighbor problem}, which we now describe.

\textbf{Locality Sensitive Hashing:} LSH was the first approximate near neighbour algorithm to break the curse of dimensionality. At their core, LSH algorithms use an \textit{LSH function} to partition the dataset into buckets. The hash function is selected so that the distance between points in the same bucket is likely to be small. To find the near neighbors of a query, we hash the query and compute the distance to every point in the corresponding bucket. Query performance can be improved with \textit{replication}, which queries multiple independent hash tables, \textit{multi-probe methods}~\cite{lv2007multi}, which examine multiple buckets in the hash table, and \textit{data-dependent LSH}~\cite{andoni2015optimal}, which tunes the hash function to the dataset. 
Recent work shows that machine learning algorithms can also construct effective LSH partitions~\cite{dong2019learning}.

\textbf{Count-Based LSH:} There are several recent algorithms which identify neighbors by counting the number of LSH collisions rather than explicitly computing distances. For example, the algorithm from~\cite{wang2018randomized} uses the count values to quickly identify potential neighbors. The algorithm from~\cite{coleman2020sub} applies compressed sensing techniques to the counts to compress the dataset, and a popular technique in genomics is to simply \textit{replace} each data point with its hash values~\cite{ondov2016mash}.

\textbf{Graphs:} Graph-based methods are another successful family of algorithms. Graph algorithms locate near neighbors by walking the edges of a graph where each point is (approximately) connected to its $k$ nearest neighbors. The focus in this area has been to improve graph properties using diversification, pruning, hierarchical structures, and other heuristics~\cite{malkov2018efficient}. Graph indices perform well on industry-standard benchmarks but are not theoretically well-understood, despite recent progress~\cite{prokhorenkova2020graph}. Graph indices also suffer from long construction times and bloated memory consumption. 


\textbf{Sample Compression:} A large number of practical methods are based on quantization. Such methods replace points in the dataset with \textit{compressed versions} of the points. Methods such as scalar quantization, vector quantization and product quantization alias each point to a collection of $k$-means centroids. One can also use machine learning to obtain learned Hamming codes for the dataset and perform efficient distance computations using bit operations. Advances in quantization are applicable to most other algorithms, but have been particularly effective when combined with brute force search on GPU hardware and partition-based search over billion-scale datasets~\cite{JDH17}.


\textbf{Group Testing:} We are not the first algorithm to apply group testing to near neighbor search. However, existing algorithms have key limitations that prevent effective practical implementations and rigorous theoretical analysis. The authors of~\cite{iscen2017memory} propose a group-based filtering algorithm based on \textit{group representative vectors}, or the vector average of group entries. To query the index,~\cite{iscen2017memory} explicitly compute the distances to all points where the distance between the representative and query exceeds a threshold. The algorithm of~\cite{shi2014group} uses the same group representatives, but applies an online backpropagation algorithm to estimate the individual similarities

This work has two shortcomings. First, the methods require many distance calculations against the group representatives (\cite{iscen2017memory} requires $N/10$ distances), resulting in poor query time. 
Second, the average vector can be similar to the query even when all points are far from the query, precluding a theoretical analysis except under restrictive distribution assumptions.
In this work, we analyze methods where \textit{only the group tests} are used to identify the neighbors, as our goal is to avoid performing expensive distance computations. Unsurprisingly, our method is theoretically and practically superior.


\section{Background}

\textbf{Formal Problem Statement: } In this paper, we solve the randomized \textit{nearest} neighbor problem. Definition~\ref{def:prob_statement} is a stronger version of the well-studied $(R,c)$-approximate \textit{near} neighbor problem. In particular, any algorithm which solves the randomized nearest neighbor problem also solves the approximate near neighbor problem with $c = 1$ and any $R\geq $ the distance to the nearest neighbor.
\begin{definition}
\label{def:prob_statement}
\textbf{Randomized Nearest neighbor:} Given a dataset $D$, a distance metric $d(\cdot,\cdot)$ and a failure probability $\delta \in [0,1]$, construct a data structure which, given a query point $y$, reports the point $x\in D$ with the smallest distance $d(x,y)$ with probability greater than $1 - \delta$.
\end{definition}

\subsection{Group Testing}


Suppose we are given a set $D$ of $N$ items, $k$ of which are positive (``hits'') and $N - k$ of which are negative (``misses''). The group testing problem is to identify the hits by grouping items and using a small collection of \textit{group tests}. A group test is positive if and only if any item from the group is positive. The objective of group testing is to reliably identify the positive items using fewer than $N$ group tests. The problem is \textit{noisy} if the tests make i.i.d. mistakes with some false positive and false negative rate. The group testing problem may also be \textit{adaptive}, where we are allowed to design test $n$ based on the results of tests $\{1,2,...n-1\}$, or \textit{non-adaptive}, where we must perform all tests at once.

Since the problem's introduction in 1943, there has been considerable work toward the construction of test designs under various constraints. For a recent review, see~\cite{aldridge2019group}. In this paper we develop near neighbor search algorithms using the noisy group testing framework. For simplicity, we mainly consider the \textit{doubly regular design}, where we evenly distribute items among $B$ tests, and we independently repeat this process $R$ times to obtain a $B\times R$ grid of group tests (Figure~\ref{fig:intuition}). However, our algorithmic framework is compatible with any non-adaptive design.

%


\subsection{Locality-Sensitive Hashing}
A hash function $h(x) \mapsto \{1,...,R\}$ is a function that maps an input $x$ to an integer in the range $[1,R]$. An LSH family $\mathcal{H}$ is a set of hash functions with the following property: Under the hash mapping, nearby points have a high probability of having the same hash value. The two points $x$ and $y$ are said to \textit{collide} if $h(x) = h(y)$. We will use the notation $s(x,y)$ to refer to the collision probability $\text{Pr}_{\mathcal{H}}[h(x) = h(y)]$. The original definition of LSH given by~\cite{indyk1998approximate} establishes lower bounds on $s(x,y)$ when $d(x,y)$ is small (i.e. we want a high probability that $x$ and $y$ collide) and upper bounds when $d(x,y)$ is large (i.e. we do not want $x$ and $y$ to collide). For our analysis, we will assume a slightly different notion of LSH. Specifically, we suppose that $s(x,y)$ is exactly equal to the similarity between $x$ and $y$. That is, $s(x,y) = \mathrm{sim}(x,y)$. The vast majority of LSH functions in the literature satisfy this property - see~\cite{gionis1999similarity} for a review.

We also introduce the concatenation trick. For any positive integer $L$, we may transform an LSH family $\mathcal{H}$ with collision probability $s(x,y)$ into a new family having $s(x,y)^L$ by sampling $L$ hash functions from $\mathcal{H}$ and concatenating the values to obtain a new hash code $[h_1(x),h_2(x),...,h_L(x)]$. If the original hash family had the range $[1,R]$, the new hash family has the range $[1, R^L]$. 


\subsection{Distance-Sensitive Bloom Filters}
The distance-sensitive Bloom filter \cite{kirsch2006distance} is a data structure which solves the \textit{approximate set membership problem}.
\begin{definition}
\label{def:approx_set_member}
\textbf{Approximate Set Membership:} Given a set $D$ of $N$ points and similarity thresholds $S_L$ and $S_H$, construct a data structure which, given a query point $y$, has:\\ 
\noindent\textbf{True Positive Rate:} If there is $x\in D$ with $\mathrm{sim}(x,y) > S_H$, the structure returns true w.p. $\geq p$ \\
\noindent\textbf{False Positive Rate:} If there is no $x \in D$ with $\mathrm{sim}(x,y) > S_L$, the structure returns true w.p. $\leq q$
\end{definition}

The distance-sensitive Bloom filter solves this problem using LSH functions and a 2D bit array. The structure consists of $m$ binary arrays that are each indexed by an LSH function. There are three parameters: the number of arrays $m$, a positive threshold $t \leq m$, and the number of concatenated hash functions $L$ used within each array. The length of each array is set to be the range of the LSH family and is therefore not a parameter.

To construct the filter, we insert elements $x \in D$ by setting the bit located at array index $[m,h_m(x)]$ to 1. To query the filter, we determine the $m$ hash values of the query $y$. If at least $t$ of the corresponding bits are set, we return true. Otherwise, we return false. For our group testing analysis, we need explicit bounds on the error rates $p$ and $q$. We obtain these bounds using a straightforward extension of Proposition 2.1 from~\cite{kirsch2006distance} and provide a proof in the supplementary materials.

\begin{thm}
\label{thm:distance_sensitive_Bloom}
Assuming the existence of an LSH family with collision probability $s(x,y) = \mathrm{sim}(x,y)$, the distance-sensitive Bloom filter solves the approximate membership query problem with
\begin{equation}
    p \geq 1 - \mathrm{exp}\left(-2m\left(-t + S_H^L\right)^{2}\right) \qquad
    q \leq \mathrm{exp}\left(-2m\left(-t + N S_L^L\right)^{2}\right)
\end{equation}
\end{thm}

\section{Algorithm}
We will now describe our algorithm for high-dimensional near neighbor search. We begin by reducing the near neighbor search problem to a group testing problem. Suppose we are given an $N$-point dataset $D$ and are asked to return points which are similar to a query $y$. If we apply a similarity threshold to the dataset, we obtain a near neighbor set $K = \{x \in D | \mathrm{sim}(x,y) \geq S\}$. We consider $K$ to be the set of ``positives'' in the group testing problem. We can solve the similarity search problem by finding the $|K|$ positives among the $N - |K|$ negatives using group testing.

In order to do so, we split the dataset $D$ into a set of groups, which we visualize as a $B\times R$ grid of cells. Each cell has a group of items $M_{r,b}$ and a corresponding group test $C_{r,b}$. To assign items to cells, we evenly distribute the $N$ points among the $B$ cells in each column of the grid, and we independently repeat this assignment process $R$ times.

To identify groups that contain positives, we need a testing procedure that outputs ``true'' when the group contains a point similar to $y$ and ``false'' otherwise. That is, we require a binary classifier $C_{r,b}$ that solves the approximate membership testing problem for $M_{r,b}$ (Definition~\ref{def:approx_set_member}). For group testing to be effective, the true positive rate $p$ and false positive rate $q$ of the classifier $C_{r,b}$ should be good enough to reliably identify positive and negative cells, respectively.

Algorithm~\ref{alg:index} shows how to construct the index. We begin by randomly distributing the points across the $B$ cells in each row, so that each cell has the same number of points. This can be done by randomly permuting the elements of $D$ and assigning blocks of $\frac{N}{B}$ elements to each cell using modulo hashing. Then, we construct classifiers (group tests) to solve the approximate membership problem in each cell.

To query the index with a point $y$, we begin by querying each classifier. If $C_{r,b}(y) = 1$, then at least one of the points in $M_{r,b}$ has high similarity to $y$. We collect all of these ``candidate points'' by taking the union of the $M_{r,b}$ sets for which $C_{r,b}(y) = 1$. We repeat this process for each of the $R$ repetitions to obtain $R$ candidate sets, one for each column in the grid. With high probability, each candidate set contains the true neighbors, but it may also have some non-neighbors that were included in $M_{r,b}$ by chance. To filter out these points, we intersect the candidate sets to obtain our approximate near neighbor set $\hat{K}$. Algorithm~\ref{alg:query} explains this process in greater detail.

\begin{minipage}{0.46\textwidth}
\begin{algorithm}[H]
\begin{algorithmic}
   
   \STATE {\bfseries Input:} Dataset $D$ of size $N$, positive integers $B$ and $R$, similarity threshold $S$
   \STATE {\bfseries Output:} A FLINNG search index consisting of membership sets $M_{r,b}$ and group tests $C_{r,b}$
  \FOR{$r=0$ {\bfseries to} $R-1$}
        \STATE Let $\pi(D)$ be a random permutation of $D$
        \STATE Define $M_{r,b} = \{\pi(D)_i \mid i \mod B = b\}$
    \ENDFOR
    \FOR{$r=0$ {\bfseries to} $R-1$}
        \FOR{$b = 0$ {\bfseries to} $B - 1$}
            \STATE Construct a classifier $C_{r,b}$ for membership set $M_{r,b}$ with true positive rate $p$ and false positive rate $q$
        \ENDFOR
    \ENDFOR
\end{algorithmic}

\caption{Index Construction}
   \label{alg:index}
\end{algorithm}
\end{minipage}
\begin{minipage}{0.46\textwidth}
\begin{algorithm}[H]
\begin{algorithmic}
   \STATE {\bfseries Input:} A FLINNG index and a query $y$
   \STATE {\bfseries Output:} Approximate set $\hat{K}$ of neighbors with similarity greater than the threshold $S$
    \STATE $\hat{K} = \{1, \ldots, N\}$
    \FOR{$r=0$ {\bfseries to} $R-1$}
        \STATE $Y = \emptyset$
        \FOR{$b=0$ {\bfseries to} $B-1$}
            \IF{$C_{r,b}(y) = 1$}
                \STATE $Y = Y \cup M_{r,b}$
            \ENDIF
        \ENDFOR
        \STATE $\hat{K} = \hat{K} \cap Y$
    \ENDFOR
\end{algorithmic}

\caption{Index Query}
   \label{alg:query}
\end{algorithm}
\vspace{1.0cm}
\end{minipage}

\begin{figure*}[t]
\begin{center}
\centerline{\includegraphics[height=1.7in]{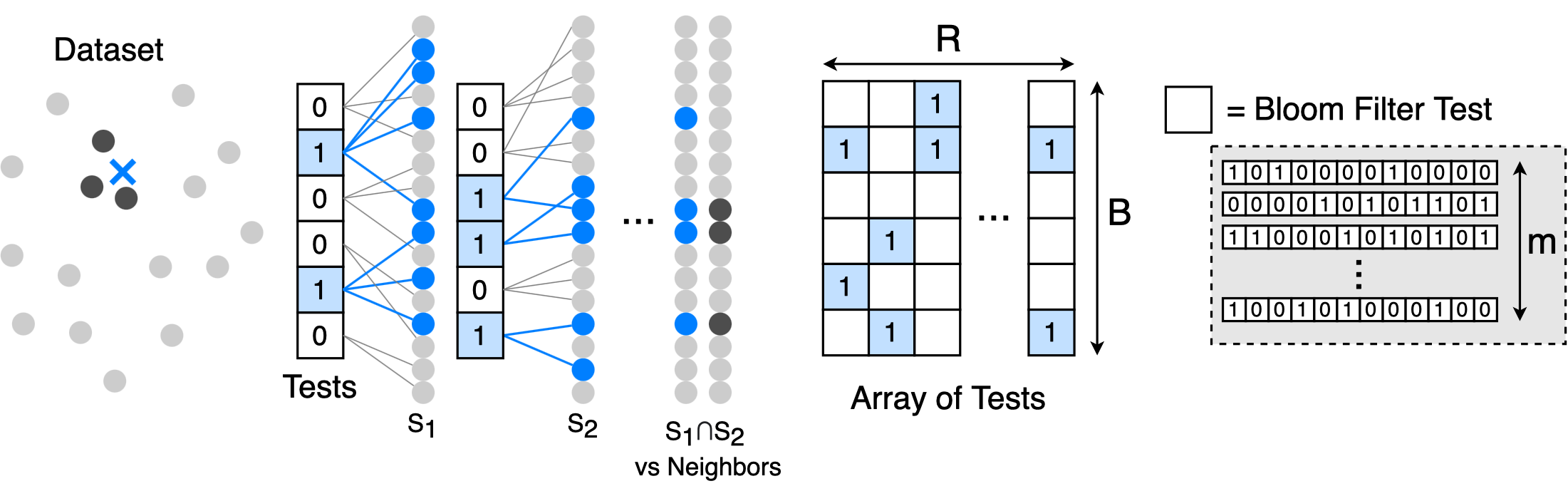}}
\vspace{-0.3cm}
\caption{Intuition behind our approach. We mark the neighbors (black dots) of a query (blue X) as ``positives'' for group testing and construct a $B\times R$ array of tests. Based on the test results, we obtain sets $S_1,S_2,...S_R$ of candidates, which we intersect to identify the neighbors. In general, the tests can be any classifier that detects neighbors. In this paper, we use distance-sensitive Bloom filters. }
\label{fig:intuition}
\end{center}
\vspace{-1cm}
\end{figure*}


\textbf{Intuition:} In each repetition, we partition $N$ points into $B$ groups, where $B \ll N$. To understand why this strategy leads to good performance, suppose we are only interested in finding the nearest neighbor $x_{\text{NN}}$ (i.e. $|K| = 1$). If the tests have a high true positive rate $p$, then the $R$ cells that contain $x_{\text{NN}}$ will have $C_{r,b}(y) = 1$. If the tests have a low false positive rate $q$, then the cells that do not contain $x_{\text{NN}}$ will have $C_{r,b}(y) = 0$ with high probability. 

In the first repetition, our tests identify the group $M_{0,b}$ that contains $x_{\text{NN}}$ - each point in $M_{0,b}$ is a near-neighbor candidate. Thus, with only $B$ calls to the classifier $C_{r,b}(y)$, we have reduced the number of candidates from $N$ to $\frac{N}{B}$. If we repeat this process, we find another candidate set $M_{1,b}$. Our overall set of candidates is now the intersection $M_{0,b}\cap M_{1,b}$, whose expected size is $\frac{N}{B^2}$. In general, each repetition reduces the number of candidates by a factor of $\frac{1}{B}$, which decreases \textit{exponentially} with the number of repetitions. We progressively rule out more and more candidates until we are left with only the near neighbors.

In practice, this process is efficient because we can construct tests with a reasonable $p$ and $q$ that are very fast to query. For example, when $C_{r,b}$ is a distance-sensitive Bloom filter, the testing process can be implemented using constant memory with bit operations or efficient integer lookup tables. The set union and intersection operations can also be implemented using cheap integer operations. The result is an algorithm that identifies the near neighbors using group testing, without explicitly storing the data or performing any distance computations.

\section{Theory}
\label{sec:theory}


\textit{Proof Sketch:}
We defer full proofs to the supplementary materials; what follows is a high-level description of our theory. To obtain theoretical guarantees, we first assume that the tests have a universal fixed false positive and false negative rate. Under this assumption, we derive bounds on the query time and error rates of the FLINNG algorithm. To satisfy this assumption, we show how to construct distance-sensitive Bloom filters with a given error rate. Here, the main technical difficulty is to bound the query time necessary to achieve the correct testing error rate. To do this, we require a data-dependent sparsity measure $\gamma$ that is small when the query has only a few close neighbors. To prove our main theorem, we set the test error rates so that the group design solves the randomized nearest neighbor problem.

\subsection{Group Testing: Runtime and Accuracy}
We first derive bounds on the error rates of our index, under assumptions about the tests.
\begin{lemma}
\label{lem:overall_tpr_fpr}
  Suppose we have a dataset $D$ of points, where a subset $K\subseteq D$ is ``positive'' and the rest are ``negative.'' Construct a $B\times R$ grid of tests, where each test has i.i.d. false positive rate $p$ and false negative rate $q$. Then Algorithm~\ref{alg:query} reports points as ``positive'' with probability:
  \begin{align}
      \mathrm{Pr}[\mathrm{Report}\, x | x \in K] \geq p^R
  \end{align}
    \begin{equation}
      \mathrm{Pr}[\mathrm{Report}\, x | x \not \in K] \leq \Bigg[q\left(\frac{eN (B-1)}{B(N-1)}\right)^{|K|} + p\left(1 - \left(\frac{N(B-1)}{eB(N-1)}\right)^{|K|}\right)\Bigg]^R
  \end{equation}
\end{lemma}  

We next bound the runtime of Algorithm \ref{alg:query}.
\begin{thm}
\label{thm:group_test_query_time}
  Under the assumptions in Lemma~\ref{lem:overall_tpr_fpr}, suppose that each test runs in time $O(T)$. Then with probability $1 - \delta$

\begin{equation}
    t_{\text{query}} = O\left(BRT + \frac{RN}{B}\left(p|K| + qB\right) \log(1/\delta)\log N\right)
\end{equation}
  
\end{thm}  

\subsection{Bounding the Test Cost}
\label{sec:bounding_test_cost}

We next bound the runtime and error rates of a specific binary classifier: a distance-sensitive Bloom filter. To distinguish between the $K$ nearest neighbors and the rest of the dataset, we apply Theorem~\ref{thm:distance_sensitive_Bloom} with $S_H = \mathrm{sim}(x_{|K|},y) = s_{|K|}$ and $S_L = \mathrm{sim}(x_{|K|+1},y) = s_{|K|+1}$, where $x_{|K|}$ is $y$'s $K$th nearest neighbor. We also assume that the filter contains $\frac{N}{B}$ points and $B = 2\sqrt{N}$. Our goal in this section is to select a threshold $t$, number of bit arrays $m$, and LSH parameter $L$ to obtain a specified value of $p$ and $q$. Once we have a test with the required error rates, we will bound the test time $T$.


Without imposing additional requirements on the query and dataset, it is impossible to design a filter for an arbitrary $p$ and $q$, as observed by~\cite{kirsch2006distance}. However, it is not a serious limitation. We can obtain the error rates provided that the query has $K$ clearly-defined neighbors and the non-neighbor points are easily distinguished from the neighbors (i.e. $s_{|K|+1}\ll s_{|K|}$). This is closely related to the stability condition from~\cite{beyer1999nearest}, so we refer to such queries as \textit{stable}. We formally define a $\gamma$-stable query as:

\begin{definition} $\gamma$\textbf{-stable Query:} We say that a query is $\gamma$-stable if 
$\frac{\log(s_{|K|})}{\log(s_{|K|+1}) - \log(s_{|K|})} \leq \gamma$
\end{definition}

We are now ready to design the classifier. Our classifier achieves the error rates $p$ and $q$ for any $\gamma$-stable query and has bounded query time.

\begin{thm}
\label{thm:bound_test_query_time}
Given a true positive rate $p$, false positive rate $q$ and stability parameter $\gamma$, it is possible to choose $m$, $L$ and $t$ so that the resulting distance-sensitive Bloom filter has false positive rate $p$ and false negative rate $q$ for all $\gamma$-stable queries. The query time is
\begin{align}
O(mL) = O\left(-\log(\min(q, 1 - p))N^{\gamma}\log(N)\right)
\end{align}
\end{thm}



\subsection{Query Time Analysis}

In this section, we combine previous results to solve the randomized nearest neighbor problem. First, we consider the query time of a $2\sqrt{N} \times R$ grid of Bloom filter classifiers. Lemma~\ref{lem:group_distance_sensitive_query_time} is a straightforward application of Theorem~\ref{thm:group_test_query_time} to the group test design from Theorem~\ref{thm:bound_test_query_time}. 


\begin{lemma}
\label{lem:group_distance_sensitive_query_time}
Under the assumptions in Lemma \ref{lem:overall_tpr_fpr}, we can use distance-sensitive Bloom filters as tests to achieve the following query time $t_{query}$ of Algorithm $2$ with probability $1 - \delta$
\begin{align*}
t_{query} = O(&RN^{\frac{1}{2} + \gamma}\log(N)\max (-\log(q), -\log(1 - p)) +RN^\frac{1}{2}\log^2(N)(|K| + qN^\frac{1}{2})\log(1/\delta))
\end{align*}
\end{lemma}

There are two ways that Algorithm~\ref{alg:query} can fail to solve the nearest neighbor problem (i.e. $|K| = 1$). We may fail to return the nearest neighbor, but we may also return any point in $D$ that is not the nearest neighbor. We can determine the values of $p$ and $q$ needed to achieve an overall failure rate $\delta$ by requiring that both events occur with probability $< \frac{\delta}{2}$ and applying the union bound.

\begin{lemma}
\label{lem:cellwise_rate_bounds}
Under the assumptions in Lemma \ref{lem:overall_tpr_fpr}, we can build a data structure that solves the randomized nearest neighbor problem for sufficiently large $N$ and small $\delta$, where\footnote{We require $N \ge 150$ and $\delta$ small enough that $R \ge 10\log N$}
\begin{align}
    p = 1 - \frac{\delta}{2R} \qquad q = N^{-\frac{1}{2}} \qquad
    R = \frac{\log(\frac{1}{\delta})}{\log (4.80N^\frac{1}{2}) - \log (2e^2 + 3.44N^\frac{1}{2})}
\end{align}
\end{lemma}

We obtain our main theorem by using the values from Lemma~\ref{lem:cellwise_rate_bounds} with the query time from Lemma~\ref{lem:group_distance_sensitive_query_time}. Note that the query time is sublinear when the (data-dependent) stability parameter $\gamma < \frac{1}{2}$.

\begin{thm} \textbf{(Main Theorem)}
\label{thm:final_query_time}
Under the assumptions of Lemma \ref{lem:cellwise_rate_bounds}, we solve the randomized nearest neighbor problem for $\gamma$-stable queries in time $t_{query}$: 
\begin{align}
t_{query} = O\left(N^{\frac{1}{2} + \gamma}\log^4(N)\log^3\left(\frac{1}{\delta}\right)\right)
\end{align}
\end{thm}

\section{Implementation}

There are several nontrivial implementation considerations to achieved good practical performance. First, we use the same $m$ LSH functions for all of the filters, allowing us to hash the query only one time. Second, we represent the distance-sensitive Bloom filters as lists of hash codes rather than bit arrays. This allows us to represent the FLINNG structure as a \textit{reverse index} from hash values to cells. The reverse index is a lookup table that, given a hash value $h$, returns a list of cells whose distance-sensitive Bloom filters contain $h$. We keep a reverse index for each of the $m$ LSH functions. 

To query the index, we use the reverse index to count the number of times that each cell collides with the query across the $m$ LSH functions. This results in an array of $B\times R$ count values, one for each cell. To obtain the classifier outputs, we mark all cells with count values larger than a threshold $t$ as ``true.'' In Theorem~\ref{thm:bound_test_query_time}, we used a global value of $t$ for all queries. However, this does not work in practice because different queries require different similarity thresholds. To address this issue, we use Algorithm~\ref{alg:relax}, which relaxes $t$ until enough cells return ``true'' so that $k$ neighbors are returned. This process is equivalent to running Algorithm~\ref{alg:query} with decreasing thresholds until $k$ points are returned.


\begin{minipage}{0.54\textwidth}
\begin{algorithm}[H]
   \caption{Threshold Relaxation Algorithm}
   \label{alg:relax}
\begin{algorithmic}
   
   \STATE {\bfseries Input:} $A$: Array of cells, sorted in descending order by hash collisions, $k$: number of neighbors to return
   \STATE {\bfseries Output:} Approximate $k$ neighbors of the query
   \STATE $\mathrm{Counts} \gets$ Array of length $N$, initialized to $0$
   \STATE $\mathrm{Result} \gets$ Empty list of IDs
   
   \FOR{$i = 0$ \bfseries{to} $B\times R - 1$}
        \FOR{point $x \in $ cell $A[i]$}
            \STATE increment $\mathrm{Counts}[x]$
            \IF{$\mathrm{Counts}[x] = R$}
                \STATE append $x$ to $\mathrm{Result}$
                \IF{$|\mathrm{Result}| = k$}
                    \STATE return $\mathrm{Result}$
                \ENDIF
            \ENDIF
        \ENDFOR
   \ENDFOR

\end{algorithmic}
\end{algorithm}
\vspace{0.2cm}
\end{minipage}
\begin{minipage}{0.44\textwidth}

\begin{figure}[H]
\begin{center}
\centerline{\includegraphics[width=2.3in]{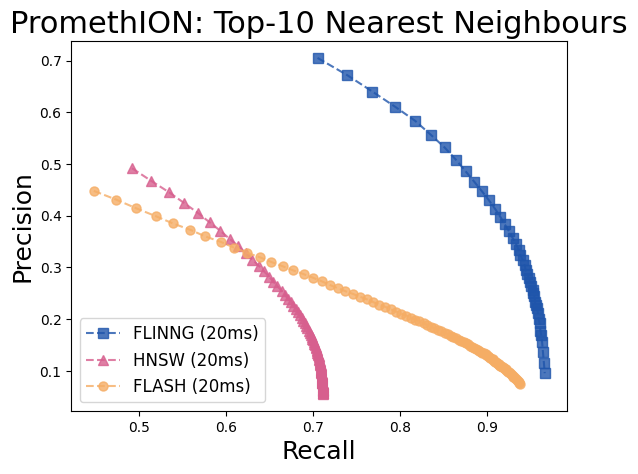}}
\caption{Precision recall tradeoff on PromethION with a query time limit of $20$ ms. Up and to the right is better.}
\label{fig:precision_recall}
\end{center}
\end{figure}
\end{minipage}

We implemented FLINNG in C++, compiled with the highest level of optimization with GCC, and used OpenMP to parallelize index construction. We implemented the reverse index as an $m \times 2^L$ table of pointers to vectors, where each vector contains a list of cells. To reduce index space, we store cell identifiers as short integers when possible. Finally, note that Algorithm~\ref{alg:relax} requires an array of length $N$ to store the count values associated with each point. For large datasets, this can exceed the CPU cache size leading to a slowdown from RAM access. However, we can avoid this issue if $R = 2$ (which is often sufficient for many applications). By storing the counts as a bit array of $N/8$ bytes, we can fit $N \leq 240$ million into a 30 MB CPU cache. We make this substitution where appropriate.

\section{Experiments}

\textbf{Datasets:} We tested FLINNG on  high-dimensional genomics, web-scale data mining, and embedding search datasets. We list the datasets in Table~\ref{tab:datasets} and briefly describe them here. RefSeqG and RefSeqP are sets of reference genome and proteome sequences for approximately 88k species~\cite{ondov2016mash}. The sequences are represented as sets of $21$-grams ($k$-mers) and are compressed via MinHash. Similarity search is relevant to RefSeq because we can answer basic scientific questions by clustering genomes. PromethION is a stream of raw metagenomic sequence reads from the latest sequencing machine by Oxford Nanopore~\cite{nicholls2019ultra}, which generates 4TB of data per day. We preprocess the reads into $16$-grams. Here, similarity search is important for read de-duplication and other pre-assembly applications. The URL dataset and Webspam datasets are from the libsvm repository. The YFCC100M dataset consists of embeddings derived from neural network activations for 100M videos and images~\cite{thomee2016yfcc100m}.

\begin{minipage}{0.62\textwidth}
\begin{table}[H]
\small\addtolength{\tabcolsep}{-1pt}
\centering
\begin{tabular}{ccccc}
\toprule
Dataset & $N$ & $d$ & $\bar{d}$ & Description\\
\midrule
RefSeqG & 117k & 1.4T & 1k & \makecell[l]{Compressed genomes} \\
\hline
RefSeqP & 117k & 1.4T & 1k & \makecell[l]{Compressed proteins} \\
\hline
PromethION & 3.7M & 4.3B & 286 & \makecell[l]{Raw sequencer data}\\ 
\hline
URL & 2.4M & 3.2M & 116 & \makecell[l]{$n$-gram features}\\
\hline
Webspam & 340k & 16.6M & 3.7k & \makecell[l]{$n$-gram features}\\ 
\hline
YFCC100M & 97M & 4096 & 4096 & \makecell[l]{Neural embeddings}\\ 
\bottomrule
\end{tabular}
\vspace{0.1cm}
\caption{Datasets: We selected data from genomics, text and embedding problems. Datasets have $N$ points and $d$ dimensions, with an average $\bar{d}$ nonzero entries per point.}
\label{tab:datasets}
\end{table}

\end{minipage}
\begin{minipage}{0.33\textwidth}
\begin{table}[H]
\small\addtolength{\tabcolsep}{-1pt}
\centering
\begin{tabular}{c|cc}
\toprule
 & Memory & Indexing \\
\midrule
FLINNG & 3.5 GB & 40 sec \\ 
FAISS & 3.7 GB & 12 hr \\ 
HNSW & >1 TB & >5 days\\ 
FLASH & 4.3 GB & 80 sec\\ 
\bottomrule
\end{tabular}
\vspace{0.1cm}
\caption{Index characteristics for YFCC100M.}

\label{tab:memory}
\end{table}

\vspace{0.95cm}
\end{minipage}

\textbf{Baselines:} We compare FLINNG against popular implementations of graph algorithms, LSH, and quantization-based search. FALCONN is a fast implementation of the traditional LSH algorithm that supports multi-probe LSH in various metric spaces~\cite{NIPS2015_2823f479}. FLASH is a recent LSH algorithm that uses aggregate LSH count statistics to avoid distance computations~\cite{wang2018randomized}. HNSW is a multi-level graph search algorithm with exceptional performance on industry-standard benchmarks~\cite{malkov2018efficient}. We use the hnswlib library and extended it to work on genomic datasets. FAISS is a highly optimized quantization-based library used for billion-scale similarity search at Facebook~\cite{JDH17}. Finally, we compare against a simple inverted index approach for sparse data when feasible, as well as our implementation (GROUPS) of the other group testing algorithm from~\cite{shi2014group}.

\textbf{Experiment setup:} We construct indices in parallel but query using a single core. Due to limitations of baseline algorithms, we were unable to evaluate all algorithms on all tasks. Due to space constraints, we provide information about hyperparameters, experiment setup, and computing hardware in the supplementary materials.

\textbf{Results:} We show the recall-query time tradeoff for all algorithms in Figure~\ref{fig:trecall}. Figure~\ref{fig:precision_recall} shows the precision-recall curve for the top 10 neighbors on the PromethION dataset when we constrain the query time to 20ms. We find that FLINNG obtains between a 2-10x speedup on many search tasks. For example, FLINNG was 3.4 times faster than FAISS at the 0.99 recall level on YFCC100M and was 4x faster than HNSW on PromethION at the 0.8 recall level. FLINNG also has a small index size and construction time when compared with baselines (Table~\ref{tab:memory}).


\begin{figure*}[t]
\begin{center}
\centerline{\includegraphics[width=\textwidth]{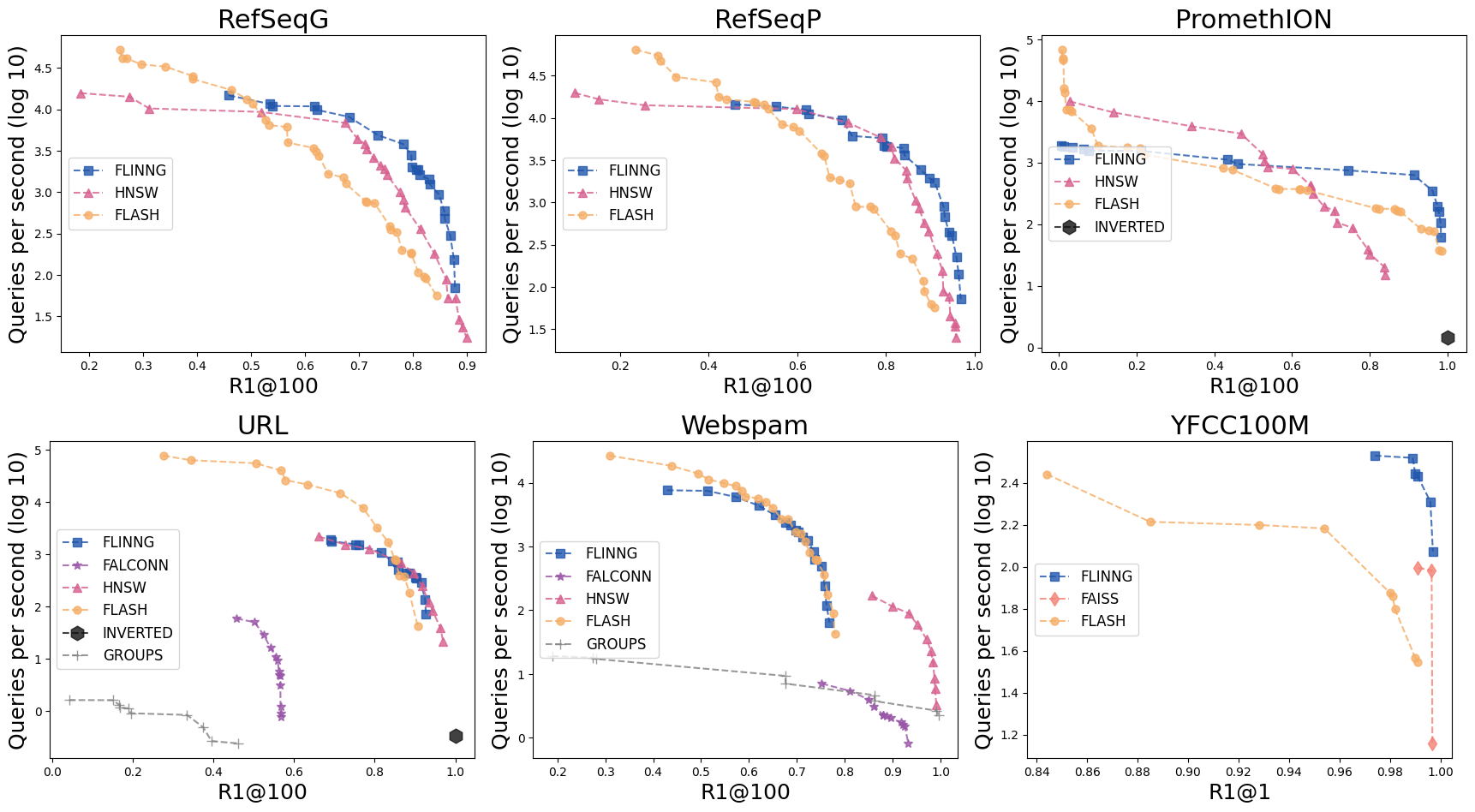}}
\caption{Time recall tradeoff for each dataset (Pareto frontier across hyperparameter configurations). Up and to the right is better. We report $\text{R}1@100$ except for YFCC100M, where we report $\text{R}1@1$.}
\label{fig:trecall}
\end{center}
\vspace{-1cm}
\end{figure*}

\section{Discussion}
\label{sec:discussion}

FLINNG is a theoretically sound algorithm with attractive practical properties that lead to a fast implementation. Our theory shows that FLINNG has sublinear query time when the query has a relatively small number of highly similar neighbors, and our experiments show that FLINNG is efficient for many real-world problems. This is particularly true in genomics, where FLINNG can outperform algorithms like HNSW by a substantial 4x margin. Existing algorithms often require expensive iterative algorithms such as $k$-means clustering or graph construction. In contrast, FLINNG relies on a simple lookup table structure that can be constructed in a single pass. 

We believe that FLINNG could be particularly effective for situations where it is hard to reduce the high-dimensional similarity search to a medium-dimensional problem. In genomics, this problem is difficult because the $k$-mer distribution contains billions of items and often has a heavy-tailed distribution. The situation may also arise for embedding tasks such as YFCC100M, where it is known that dimensionality reduction can hurt performance. In these scenarios, where we prefer to search over the original metric space, FLINNG provides a fast and scalable solution.

\textbf{Limitations and Ethical Considerations:} We do not foresee any ethical problems. The main limitation of our method is that it works best on high-dimensional search tasks where the neighbors are all above a (relatively high) similarity threshold. For this reason, FLINNG is likely not the best choice for problems such as $k$-NN classification, where low-similarity results may be important.

\bibliographystyle{plain}
\bibliography{references}

\appendix
\setcounter{thm}{1}
\setcounter{lemma}{1}
\setcounter{definition}{1}

\section*{Supplementary Materials}

\section{Proofs}

In this section, we provide proofs for all of the theorems introduced in the main text. We begin with a simple extension of the results of \cite{kirsch2006distance} for the Bloom filter false positive and negative rates. Then, we prove our main claim, which is that the query time of our data structure is sublinear, given some relatively weak assumptions on the stability of the query. 

\begin{thm}
\label{thm:distance_sensitive_Bloom}
Assuming the existence of an LSH family with collision probability $s(x,y) = \mathrm{sim}(x,y)$, the distance-sensitive Bloom filter solves the approximate membership query problem with
\begin{equation}
    p \geq 1 - \mathrm{exp}\left(-2m\left(-t/m + S_H^L\right)^{2}\right)
\end{equation}
\begin{equation}
    q \leq \mathrm{exp}\left(-2m\left(-t/m + N S_L^L\right)^{2}\right)
\end{equation}
\end{thm}
\begin{proof} We begin with a brief explanation of the results from \cite{kirsch2006distance}. Recall that a distance-sensitive Bloom filter is a collection of $m$ bit arrays. Array $i$ is indexed using an independent LSH function $l_i(x)$. To insert a point $x$ into the $i$\textsuperscript{th} array, we set the bit at location $l_i(x)$ to `1.' To query the filter, we calculate the $m$ hash values of the query and return ``true'' when at least $t$ of the corresponding bits are `1.'

To bound $p$ (the true positive rate) and $q$ (the false positive rate), we bound the probability that a single array returns ``true.'' Since the arrays are independent, the number of `1's follows a Binomial distribution. \cite{kirsch2006distance} obtain their main result (Theorem 3.1, in their paper) using the Azuma-Hoeffding inequality to bound the tail of the Binomial distribution. This is done for the Hamming metric, using a specially-chosen value of $t$. We repeat their analysis, but for any value of $t$ and in a general metric space.

\textbf{True Positive Rate:} First, we will prove the bound for $p$, the true positive rate. Given a query $y$, let $p_y$ be the probability
$$p_y = \mathrm{Pr}[\mathrm{bit\:at\:} h(y) = 1]$$

From Proposition 2.1 of \cite{kirsch2006distance}, we have that
$$ p_y \geq \max_{x\in D} \mathrm{sim}(x,y)^L > S_H^L$$
Note that $\max \mathrm{sim}(x,y) \geq S_H$ because there is a point $x \in D$ with $\mathrm{sim}(x,y) \geq S_H$, as $y$ is a true positive. Also, note that the $\frac{Nm}{\mathrm{len}}$ term (where $\mathrm{len}$ is the length of the array) is zero in our case, as we do not perform rehashing. Finally, note that we concatenate $L$ hashes so that the LSH function has collision probability $\mathrm{sim}(x,y)^L$.

We use the Azuma-Hoeffding inequality as in \cite{kirsch2006distance}. Let $B(q_y) = \mathrm{Binom}(m,q_y)$ and observe that
$$ t - \mathbb{E}[B(q_y)] = t - mq_y \leq t - mS_H^L$$
\begin{align}
    1 - p &= \mathrm{Pr}[B(q_y) < t]\\
    & = \mathrm{Pr}[B(q_y) - \mathbb{E}[B(q_y)] < t - B(q_y)]\\
    & \leq \mathrm{Pr}[B(q_y) - \mathbb{E}[B(q_y)] < t - mS_H^L] \\
    & \leq \mathrm{exp}\left(-\frac{2}{m}(t - mS_H^L)^2\right)
\end{align}
Finally, we have
$$ p \geq 1 - \mathrm{exp}\left(-2m(-t/m + S_H^L\right)$$

\textbf{False Positive Rate:} The false positive analysis is similar. Here, we have (again, from Proposition 2.1 of \cite{kirsch2006distance}) that
$$ p_y \leq \sum_{x\in D} \mathrm{sim}(x,y)^L \leq N \max_{x\in D}\mathrm{sim}(x,y)^L = N S_L^L$$
We will use the Azuma-Hoeffding inequality again, this time using the fact that
$$ t - \mathbb{E}[B(p_y)] \geq t - mN S_L^L$$
\begin{align}
    q &= \mathrm{Pr}[B(q_y) \geq t]\\
    & = \mathrm{Pr}[B(q_y) - \mathbb{E}[B(q_y)] \geq t - B(q_y)]\\
    & \leq \mathrm{Pr}[B(q_y) - \mathbb{E}[B(q_y)] \geq t - mNS_L^L] \\
    & \leq \mathrm{exp}\left(-\frac{2}{m}(t - mNS_L^L)^2\right)
\end{align}
Leaving us with the desired inequality: 
$$ q \leq \mathrm{exp}\left(-2m(-t/m + S_L^L\right)$$
\end{proof}


\subsection{Group Testing: Runtime and Accuracy}
The first step of our analysis is to bound the true positive rate and false positive rate of our group testing design. We suppose that the individual tests have a true positive rate $p$ and a false positive rate $q$, and that the tests are independent. Under these assumptions, we have the following result:
\begin{lemma}
\label{lem:overall_tpr_fpr}
  Suppose we have a dataset $D$ of points, where a subset $K\subseteq D$ is ``positive'' and the rest are ``negative.'' Construct a $B\times R$ grid of tests, where each test has i.i.d. true positive rate $p$ and false positive rate $q$. Then Algorithm 2 reports points as ``positive'' with probability: 
  \begin{align}
      \mathrm{Pr}[\mathrm{Report}\, x | x \in K] \geq p^R
  \end{align}
    \begin{equation}
      \mathrm{Pr}[\mathrm{Report}\, x | x \not \in K] \leq \Bigg[q\left(\frac{eN (B-1)}{B(N-1)}\right)^{|K|} + p\left(1 - \left(\frac{N(B-1)}{eB(N-1)}\right)^{|K|}\right)\Bigg]^R
  \end{equation}
\end{lemma}
\begin{proof} The procedure returns ``positive'' for an element $x \in D$ only if all of the tests whose group contains $x$ return ``positive.'' Therefore, we will analyze the probability of reporting $x\in D$ as ``positive'' in each repetition (column) of the grid. We obtain the final probabilities using the fact that the repetitions are independent.

\textbf{True Positive Rate:} 
If $x\in K$, then each group containing $x$ returns ``positive'' with the true positive rate $p$. Therefore,
$$\mathrm{Pr}[\mathrm{Report}\, x | x \in K] = p^R $$

\textbf{False Positive Rate:} If $x\not\in K$, then there are two ways for the repetition to accidentally report $x$ as ``positive.'' The first way is for $x$ to fall into a group that contains one of the $|K|$ true positives, and for this test to correctly report ``positive'' (which happens with probability $p$). The second way is for $x$ to fall into a group that exclusively contains negatives, but for the test to malfunction (which occurs with probability $q$). 

Let $p_x$ be the probability that we report $x$ in a single repetition. That is,
$$\mathrm{Pr}[\mathrm{Report}\, x | x \not\in K] = p_x^R $$
$$= p \mathrm{Pr}[x \in \mathrm{positive\,group}] + q\mathrm{Pr}[x \in \mathrm{negative\,group}]$$

The probability that $x$ falls into a negative group is determined by the hypergeometric distribution. Each group contains $N/B$ points\footnote{We suppose that $B$ evenly divides $N$ for simplicity. We may accommodate the general case by replacing $N$ in our inequality with $N+1$, which does not asymptotically change our results.} from the dataset, which we draw from a pool of $N-|K|$ negatives and $|K|$ positives. Since we condition on $x$ being negative, we draw $N/B - 1$ points from $N-1$ possibilities, $|K|$ of which are positive. Therefore, the probability that $x$ falls into a negative group is equal to the probability mass function of the hypergeometric distribution, evaluated at zero.

$$ \mathrm{Pr}[x \in \mathrm{negative\,group}] = \frac{ \binom{|K|}{0}\binom{N-|K|-1}{N/B-1} }{ \binom{N-1}{N/B-1} } $$
Which, after simplification and using Vandermonde's identity:
$$ \mathrm{Pr}[x \in \mathrm{negative\,group}] = \frac{\binom{N-N/B}{|K|}}{\binom{N-1}{|K|}} $$
$$ \mathrm{Pr}[x \in \mathrm{positive\,group}] = 1 - \mathrm{Pr}[x \in \mathrm{negative\,group}]$$

We wish to have an upper bound on both quantities, which amounts to having both an upper and lower bound on $\mathrm{Pr}[x \in \mathrm{negative\,group}]$. We repeatedly apply the following inequalities
$$ \left(\frac{a}{b}\right)^b \leq \binom{a}{b} \leq \left(\frac{ea}{b}\right)^b $$
to get
$$ \frac{\left(\frac{N-N/B}{|K|}\right)^{|K|}}{\left(e\frac{N-1}{|K|}\right)^{|K|}} \leq \frac{\binom{N-N/B}{|K|}}{\binom{N-1}{|K|}} \leq \frac{\left(e\frac{N - N/B}{|K|}\right)^{|K|}}{\left(\frac{N-1}{|K|}\right)^{|K|}} $$
or, after simplification
$$ \left(\frac{N(B-1)}{eB(N-1)}\right)^{|K|} \leq \frac{\binom{N-N/B}{|K|}}{\binom{N-1}{|K|}} \leq \left(\frac{eN(B-1)}{B(N-1)}\right)^{|K|} $$
This yields the inequalities
$$ \mathrm{Pr}[x \in \mathrm{negative\,group}] \leq \left(\frac{eN(B-1)}{B(N-1)}\right)^{|K|}$$
$$ \mathrm{Pr}[x \in \mathrm{positive\,group}] \leq 1 - \left(\frac{N(B-1)}{eB(N-1)}\right)^{|K|}$$
which, when substituted into the expression for $p_x^R$, proves the theorem.
\end{proof}

We next bound the runtime of our method.
\begin{thm}
\label{thm:group_test_query_time}
  Under the assumptions in Lemma~\ref{lem:overall_tpr_fpr}, suppose that each test runs in time $O(T)$. Then with probability $1 - \delta$
  
\begin{equation}
    t_{\text{query}} = O\left(BRT + \frac{RN}{B}\left(p|K| + qB\right) \log(1/\delta)\log N\right)
\end{equation}  
\end{thm}
\begin{proof}
We must query each group test, and then intersect all the candidate groups. The $BRT$ term is the cost of querying all $B\times R$ cells. To obtain the cost of intersecting the $R$ candidate sets, let $P_i$ be the candidate set of the $i$\textsuperscript{th} repetition and let $c_{i}$ be the cost of the $i$\textsuperscript{th} intersection, where $i = \{1,2,...,R\}$. The total cost is
$$ c_{\text{total}} = \sum_{i = 1}^{R-1} O\left(|P_{i+1}| + \left|\bigcap_{j=1}^i P_{j}\right|\right) + O\left(|P_{i=1}| \log |P_{i=1}|  + \left|\bigcap_{j=1}^i P_{j}\right| \log \left|\bigcap_{j=1}^i P_{j}\right|\right)$$
because the cost to intersect two sets of sorted integers is the sum of set cardinalities, and we pay an $O(|P|\log |P|)$ cost to sort a list of size $|P|$. Also, note that 

$$\left|\bigcap_{j=1}^i P_{j}\right| \leq |P_i| $$
because $|A\cap B|\leq \min\{|A|,|B|\}\leq |A|$. Therefore, we have
$$ c_{\text{total}} = \sum_{i = 1}^{R-1} |P_{i}| +|P_{i+1}| + O(N\log N) = O\left(RN \log N + \sum_{i = 1}^R|P_i|\right)$$
Because each group has exactly $\frac{N}{B}$ points, the value of $|P_i|$ is $\sum_{j = 1}^B \frac{N}{B}\mathbbm{1}_{(i,j)}$, where the indicator function $\mathbbm{1}_{(i,j)} = 1$ if the group test in row $j$ and column $i$ outputs ``positive.'' Under the mild assumption that $p > q$ (i.e. the true positive rate is larger than the false positive rate), this sum is maximized when all $|K|$ true positives are assigned to different groups. The expected value of this sum is
$$ \mu = \mathbb{E}[|P_i|] \leq p|K| + q(B - |K|) \leq p|K| + qB$$

We want to bound this sum in probability:

$$ \mathrm{Pr}\left[\sum_{j=1}^B \mathbbm{1}_{(i,j)} \geq (1+\Delta) \mu\right]$$

We use the simplified Chernoff bound for independent non-identical Bernoulli sums: 
$$ \mathrm{Pr}\left[\sum_{j=1}^B \mathbbm{1}_{(i,j)} \geq (1+\Delta) \mu\right] \leq e^{-\frac{\Delta^2\mu}{2+\Delta}}$$

We wish to find the value of $\Delta$ which makes this probability smaller than the failure rate.
$$ e^{-\frac{\Delta^2\mu}{2+\Delta}} < \delta$$
$$ \frac{\Delta^2 \mu}{2+\Delta} \geq \log 1/\delta $$

In our context, $\mu > 1$ (otherwise, it is trivial to bound $|P_i|$) and we may constrain $\Delta > 1$. This yields the inequality
$$ \frac{\Delta^2 \mu}{2+\Delta} \geq \frac{\Delta^2}{3\Delta}\log 1/\delta $$
Therefore, we may set $\Delta = 3\log 1/\delta$ to get the following statement with probability $1 - \delta$
$$\sum_{j=1}^B \mathbbm{1}_{(i,j)} < (1+3\log(1/\delta))\left(p|K| + qB \right) $$

This bound on the intersection cost proves the theorem.

\end{proof}

\subsection{Bounding the Test Cost}
In this section, we bound the runtime and error characteristics of a specific binary classifier: a distance-sensitive Bloom filter. To distinguish between the $K$ nearest neighbors and the rest of the dataset, we apply Theorem~\ref{thm:distance_sensitive_Bloom} to a group with $\frac{N}{B}$ points, $S_H = \mathrm{sim}(x_{|K|},y) = s_{|K|}$, and $S_L = \mathrm{sim}(x_{|K|+1},y) = s_{|K|+1}$, where $x_{|K|}$ is the $K$ nearest neighbor to the query $y$. This gives us the following bounds on $q$ and $p$:
\begin{align}
     p &\ge 1 - exp\left(-2m\left(-\frac{t}{m} + \left(s_{|K|}\right)^L\right)^2\right) \label{eq:p_lower} \\
     q &\le \exp\left(-2m\left(- \frac{t}{m} + \frac{N}{B}\left(s_{|K| + 1}\right)^L\right)^2\right) \label{eq:q_upper} 
\end{align} 
These bounds have four design parameters that we may freely choose: $t$, the threshold number of collisions we require to report a ``positive''; $m$, the number of bit arrays in the Bloom filter; $L$, the number of hash values we concatenate together in each array; and $B$, the number of cells into which the dataset is partitioned within each column. In this section, we seek to find specific values for these free parameters that will allow us to build a distance sensitive Bloom filter with sufficiently high $p$ and low $q$. We will use these values in the proof of Theorem~\ref{thm:bound_test_query_time} and for the rest of our analysis. 

\textbf{A note about the hashing cost:} The cost to query each filter is the cost of performing $L\times m$ LSH computations. However, the LSH computations are not $O(1)$, they are $O(d)$, where $d$ is the dimensionality of the dataset. Since this simply adds a constant multiplier term of $d$ to the asymptotic expressions, we not not include the dependency on $d$ in our analysis.

\textbf{Choosing a Value For $t$:} In the bounds given in Equation~\eqref{eq:p_lower} and Equation~\eqref{eq:q_upper}, the inner expression with $-\frac{t}{m}$ is squared, which gives the initial impression that any value of $t$ will work. However, when we derived these bounds in Theorem~\ref{thm:distance_sensitive_Bloom}, we implicitly require that the threshold ratio (here $\frac{t}{m}$) be smaller than the ``positive'' Bloom filter collision probability $S_H^L$ and larger than the ``negative'' Bloom filter collision probability $NS_L^L$. Making the same substitutions, we find that $\frac{t}{m}$ must satisfy the following condition.
\begin{equation}
    \frac{N}{B}(s_{|K| + 1})^L < \frac{t}{m} < (s_{|K|})^L \label{eq:t_bound}
\end{equation}
The intuition behind this condition is that $t$ must be a threshold number of collisions that lies between the expected number of collisions in a positive group and the expected number of collisions in a negative one.

Any $\frac{t}{m}$ in this range gives us \textit{some} valid bound on $p$ and $q$, but values very close to the edges of the range are suboptimal. Rather than find an optimal value of $t$, which is likely difficult and data-dependent, we choose a specific value for $\frac{t}{m}$ that works well: the average between the lower bound and the upper bound in Equation~\eqref{eq:t_bound}. This gives us the following value for $t$:
\begin{equation}
    t = m\left(\frac{\frac{N}{B}(s_{|K| + 1})^L + (s_{|K|})^L}{2}\right) \label{eq:t_def}
\end{equation}
One benefit of choosing this value for $t$ is that the bounds on $p$ and $q$ from Equation~\eqref{eq:p_lower} and Equations~\eqref{eq:q_upper} now look the same. After substituting $t$ into the bounds, we define a new variable $\alpha$ as
\begin{equation}
    \alpha = \exp\left(-2m\left(\frac{(s_{|K|})^L - \frac{N}{B}(s_{|K| + 1})^L}{2}\right)^2\right) \label{eq:a_def}
\end{equation}
Note that
\begin{equation}
    q \le \alpha \text{ and } p \ge 1 - \alpha \label{eq:pqa}
\end{equation}
$\alpha$ simplifies the analysis because it represents the bounds on $p$ and $q$ at the same time. If we decrease $\alpha$, we have a larger $p$ and a smaller $q$ (i.e. a more accurate test). 

\textbf{Choosing a Value For $B$ and $L$:} To further simplify the analysis, we wish to decouple $\alpha$ from $N$ in Equation~\eqref{eq:a_def}.
Our goal in this section is to choose $B$ and $L$ as a function of $N$ so that the error rate $\alpha$ no longer depends on $N$.
As with $t$, our choices are not necessarily optimal. We use them because they allow us to prove theoretical guarantees about the system.

We first let
\begin{equation}
    B = 2\sqrt{N} \label{eq:b_def}
\end{equation}
To have sublinear query time, $B$ must be proportional to some fractional power of  $N$ because the query time in  Theorem~\ref{thm:group_test_query_time} contains both $\frac{N}{B}$ and $B$ factors. The use of $N^{\frac{1}{2}}$ minimizes the complexity of their sum, and the constant factor of $2$ is chosen to simplify the analysis in the next paragraph.

We next let $L$ be the smallest positive integer such that $(s_{|K|})^L \geq 2\frac{N}{B}(s_{|K| + 1})^L$. Since $s_{|K|} \geq s_{|K|+1}$, it is always possible to find such an integer (even though this integer may be impractically large). This choice of $L$ simplifies the difference in the squared term in $\alpha$ to $\frac{(s_{|K|})^L}{2}$. In particular, $\alpha$ no longer depends on $N$. The analysis now depends on $N$ exclusively through the parameter $L$. 

If we start with $\frac{N}{B}(s_{K + 1})^L = \frac{(s_{K})^L}{2}$ and solve for L, as well as plug in our expression for $B$ from Equation~\eqref{eq:b_def}, we get the following expression for L.

\begin{equation}
L = \frac{\frac{1}{2}\log(N)}{\log(s_{|K|}) - \log(s_{|K| + 1})} \label{eq:l_def}
\end{equation}

Finally, we can plug this value of $L$ into Equation~\eqref{eq:a_def} and simplify, which gives us the following value for $\alpha$.
\begin{equation}
    \alpha = \exp\left(\frac{-m(s_{|K|})^{\frac{\log(N)}{\log (s_{|K|}) - \log(s_{|K| + 1})}}}{8}\right) \label{eq:a_2}
\end{equation}
We will use the following well-known fact: for any non-negative and nonzero real numbers $a, b, c$,
\begin{equation}
    a^{\frac{\log(b)}{c}} = b^{\frac{\log(a)}{c}} \label{eq:clever}
\end{equation}
We apply Equation~\eqref{eq:clever} to Equation~\eqref{eq:a_2}, with $a = s_{|K|}$, $b = N$, and $c = \log (s_{|K|}) - \log(s_{|K| + 1})$ and simplify the result to obtain
\begin{equation}
    \alpha = \exp\left(\frac{-mN^{\frac{\log(s_{K})}{\log (s_{|K|}) - \log(s_{|K| + 1})}}}{8}\right) \label{eq:a_3}
\end{equation}

\textbf{The $\gamma$-Stable Query Condition:}

As observed by \cite{kirsch2006distance}, it is not possible to design a distance-sensitive Bloom filter with arbitrary $p$ and $q$ without additional conditions on the query. Therefore, we introduce the requirement that the query be $\gamma$-\textit{stable}. That is, we require that the query not be a pathologically difficult query for a distance-sensitive Bloom filter to answer\footnote{This parameter functions similarly to $\rho$ from the standard LSH near neighbor analysis.}. When the $\gamma$ is small, the query has $|K|$ neighbours with a clear distinction between non-neighboring points (i.e. $s_{|K|+1}\ll s_{|K|}$). When $|K|$ is large, there is no such distinction, and the neighboring points are approximately as close as the neighbors. Unstable queries are both rare and non-informative in the near neighbor setting - for further discussion, see the seminal paper of \cite{beyer1999nearest}. We formally define a $\gamma$-stable query as:

\begin{definition} $\gamma$\textbf{-stable Query:} We say that a query is $\gamma$-stable if
\label{def:gamma}
\begin{equation}
    \frac{\log(s_{|K|})}{\log(s_{|K| + 1}) - \log(s_{|K|})} \leq \gamma
\end{equation}
\end{definition}

First we note that similarity is a measure from $0$ to $1$, so the numerator of $\gamma$ is negative. Furthermore, since $s_{|K| + 1} < s_{|K|}$ and the $\log$ function is monotonically increasing, $s_{|K| + 1} - s_{|K|} < 0$ and the denominator is negative. Thus, $\gamma$ is positive (a negative number over a negative number). Indeed, if $\gamma$ is small then the query is stable because $\gamma$ is small when the similarity of $x$ with the $|K|$\textsuperscript{th} nearest neighbor is large (the numerator is a small negative number) and when there is a large similarity gap between the $|K|$\textsuperscript{th} and $|K + 1|$\textsuperscript{th} neighbors (the denominator is a large negative number). 

Our $\gamma$-stable parameterization has the added benefit of removing the similarity values $s_{|K|}$ and $s_{|K| + 1}$ from Equation~\eqref{eq:a_3}. Substituting the query parameterization from Definition~\ref{def:gamma} to Equation~\eqref{eq:a_3}, we have the following upper bound for $\alpha$. This bound holds for all $\gamma$-stable queries:
\begin{equation}
    \alpha \le \exp\left(\frac{-mN^{-\gamma}}{8}\right)  \label{eq:a_final}
\end{equation}

The reason Equation~\eqref{eq:a_final} is an upper bound is because Definition~\ref{def:gamma} guarantees that $-\gamma$ is a \textit{larger in magnitude} negative number than the negative exponent of $N$ in Equation~\eqref{eq:a_3}. Thus, when we raise $N$ to these exponents, we have that 
\begin{equation}
     N^{-\gamma} < N^{\frac{\log(s_{|K|})}{\log(s_{|K|}) - \log(s_{|K|+1})}}
\end{equation}

Thus, our use of $\gamma$ results in a \textit{smaller} (in magnitude) negative number inside the exponential in Equation~\eqref{eq:a_3}, and thus a larger $\alpha$ overall. Note that this somewhat loosens our bounds for $p$ and $q$ (since they are in terms of $\alpha$).

Given these parameter choices, we are ready to state our theorem which bounds the query time of each distance-sensitive Bloom filter. We use $p'$ and $q'$ to denote the desired true positive rate and false positive rate of the Bloom filter, to differentiate from $p$ and $q$ above.

\begin{thm}
\label{thm:bound_test_query_time}
Given a true positive rate $p'$, false positive rate $q'$ and stability parameter $\gamma$, it is possible to choose $m$, $L$ and $t$ so that the resulting distance-sensitive Bloom filter has true positive rate $p'$ and false positive rate $q'$ for all $\gamma$-stable queries. The query time is
\begin{align}
O(mL) = O\left(-\log(\min(q', 1 - p'))N^{\gamma}\log(N)\right)
\end{align}
\end{thm}
\begin{proof}

From Equation~\eqref{eq:pqa}, we have $q \le \alpha$ and $p \ge 1 - \alpha$. To guarantee that the actual error rates of the filter meet the design requirements (i.e. $q \le q'$ and $p \ge p'$), we must choose $\alpha$ small enough such that $q' \le \alpha$ and $p' \ge 1 - \alpha$. Thus we need
\begin{equation}
     \alpha \le \min(q', 1 - p') \label{eq:a_nec}
\end{equation}
For $\gamma$-stable queries, we may bound $\alpha$ using Equation~\eqref{eq:a_final}. Because it is more expensive to design a filter with small $\alpha$, we wish to use the largest possible value of $\alpha$ that will work. This value is attained when the upper bound in Equation~\eqref{eq:a_final} is equal to the upper bound in Equation~\eqref{eq:a_nec}. This gives us the following condition for the Bloom filter to have the desired error characteristics:
\begin{equation}
    \exp\left(\frac{-mN^{-\gamma}}{8}\right) = \min(q, 1 - p)
\end{equation}
Simplifying, we have the following expression for $m$. Note that $m$ is positive because $\log(\min(q,1-p)) < 0$. 
\begin{equation}
    m = -8\log(\min(q, 1 - p))N^\gamma \label{eq:m_def}
\end{equation}
We now have concrete values for all of our free variables: $t$ from Equation~\eqref{eq:t_def}, $B$ from Equation~\eqref{eq:b_def}, $L$ from Equation~\eqref{eq:l_def}, and $m$ from Equation~\eqref{eq:m_def}. Note that $L$ should be chosen as the maximum over all $s_{|K|}$ and $s_{|K|+1}$ that are $\gamma$-stable\footnote{This requires a mild assumption that the ratio $\frac{s_{|K|}}{s_{|K|+1}}$ is bounded.}, so that Equation~\eqref{eq:l_def} is true for \textit{all} $\gamma$-stable queries.

When $t$,$B$,$L$, and $m$ are chosen in this way, the resulting filter has true positive rate $p \leq p'$ and false positive rate $q \leq q'$ for all $\gamma$-stable queries.

To query a distance-sensitive Bloom filter, we must compute $m\times L$ hash functions and perform $m$ array lookups for an overall complexity of $O(mL)$. Thus, the time to query our Bloom filter is
\begin{align*}
    O(mL) &= O\left(\log(\min(q, 1 - p))N^\gamma  \log(N)\right)
\end{align*}
\end{proof}



\subsection{Query Time Analysis}

In this section, we combine previous results to solve the randomized nearest neighbor problem. First, we consider the query time of a $2\sqrt{N} \times R$ grid of Bloom filter classifiers. Lemma~\ref{lem:group_distance_sensitive_query_time} is a straightforward application of Theorem~\ref{thm:group_test_query_time} to the group test design from Theorem~\ref{thm:bound_test_query_time}. 


\begin{lemma}
\label{lem:group_distance_sensitive_query_time}
Under the assumptions in Lemma \ref{lem:overall_tpr_fpr}, we can use distance-sensitive Bloom filters as tests to achieve the following query time $t_{query}$ of Algorithm $2$ with probability $1 - \delta$
\begin{align}
\begin{split}
t_{query} = O(&RN^{\frac{1}{2} + \gamma}\log(N)\max (-\log(q), -\log(1 - p))\\ + &RN^\frac{1}{2}\log^2(N)(|K| + qN^\frac{1}{2})\log(1/\delta))
\end{split}
\end{align}
\end{lemma}
\begin{proof}
    
Each cell is a distance-sensitive Bloom filter, so the test query time $T$ is equal to our result for the query time of a distance-sensitive Bloom filter from Theorem~\ref{thm:bound_test_query_time}.
    
We use $B = 2\sqrt{N}$ from Equation~\eqref{eq:b_def} and $O(T) = O\left(\log(\min(q, 1 - p))N^\gamma \log(N)\right)$ from Theorem~\ref{thm:bound_test_query_time} with the query time expression from Theorem~\ref{thm:group_test_query_time} to get 
\begin{align}
\begin{split}
    t_{query} = O(&-RN^{\frac{1}{2} + \gamma}\log(N)\log (\min(q, 1 - p))\\ + &RN^\frac{1}{2}\log^2(N)(p|K| + qN^\frac{1}{2})\log(1/\delta)) \label{eq:expected_initial}
\end{split}
\end{align}
Since $p < 1$, we may replace it with $1$. Also, since $0 < p < 1$ and $0 < q < 1$,
\begin{align}
    -\log(\min(p , q)) = \max(-\log(p), -\log(q))
\end{align}
Making these two substitutions into Equation~\eqref{eq:expected_initial}, we have
\begin{align}
\begin{split}
t_{query}= O(&RN^{\frac{1}{2} + \gamma}\log(N)\max (-\log(q), -\log(1 - p))\\ + &RN^\frac{1}{2}\log^2(N)(|K| + qN^\frac{1}{2})\log(1/\delta))
\end{split}
\end{align}
\end{proof}

Our bound on the query time has two free parameters: $p$ and $q$. We will show that, given an appropriate choice for $p$ and $q$, our algorithm solves the nearest neighbor problem (i.e. $|K| = 1$) in sublinear time. We begin by finding the values of $p$ and $q$ which solve the nearest neighbor problem in Lemma~\ref{lem:cellwise_rate_bounds}. 

\begin{lemma}
\label{lem:cellwise_rate_bounds}
Under the conditions in Lemma \ref{lem:overall_tpr_fpr}, our data structure solves the randomized nearest neighbor problem for sufficiently large $N$ and small $\delta$, where\footnote{We require $N \ge 150$ and $\delta$ small enough that $R \ge 10\log N$}
\begin{align}
    p &= 1 - \frac{\delta}{2R} \qquad q = N^{-\frac{1}{2}}\\
    R &= \frac{\log(\frac{1}{\delta})}{\log (4.80N^\frac{1}{2}) - \log (2e^2 + 3.44N^\frac{1}{2})}
\end{align}
\end{lemma}
\begin{proof}

For this proof we will use the index described in Algorithm 1, using $R$ columns of $B = 2\sqrt{N}$ distance-sensitive Bloom filter cells. We now will determine the requirements for $p$, $q$, and $R$ to achieve an overall failure rate of $\delta$ in Algorithm 2.

There are two ways that the querying process, Algorithm 2, can fail to solve the nearest neighbor problem. We may fail to return the nearest neighbor, but we may also return any point in $D$ that is not the nearest neighbor. Let $P$ be the probability that Algorithm 2 returns the nearest neighbor $x_{\text{NN}}$ to the query, and let $Q$ be the probability Algorithm 2 returns at least one other point in $D$. Then using the union bound, we solve the nearest neighbor problem if
\begin{equation}
    (1 - P) + Q < \delta
\end{equation}
To simplify our analysis, we constrain $(1 - P)$ and $Q$ to be less than $\frac{\delta}{2}$. 
\begin{equation}
    (1 - P) \le \frac{\delta}{2} \qquad Q < \frac{\delta}{2} \label{eq:half_bound}
\end{equation}
\textbf{Analysis of $1 - P$:} Since $|K| = 1$ for the nearest neighbor problem, the true positive rate from Theorem~\ref{thm:distance_sensitive_Bloom} is equal to $P$, so that $P \ge p^R$. Thus $1 - P \le 1 - p^R$, so we will achieve the necessary bound on $P$ in Equation~\ref{eq:half_bound} if
\begin{equation}
    1 - p^R \le \frac{\delta}{2}
\end{equation}
Rearranging the inequality in terms of $1-p$, we have a constraint on the values of $p$ which solve the nearest neighbor problem.
\begin{equation}
    1 - p \le 1 - \left(1 - \frac{\delta}{2}\right)^\frac{1}{R} \label{eq:p_bound_1}
\end{equation}
If $p$ satisfies the inequality, then $1 - P < \frac{\delta}{2}$). Thus, we may \textit{reduce} the right hand side of the inequality, and the resulting values of $p$ will \textit{still} solve the nearest neighbor problem.

We now tighten the constraint in Equation~\eqref{eq:p_bound_1}. We use the Generalized Bernoulli's inequality, which states that
\begin{equation}
    (1 - x)^r \le 1 - rx \qquad \text{for $r \in [0, 1]$} \label{eq:bernoulli}
\end{equation}
Since $R \ge 1$, $\frac{1}{R} \in [0, 1]$, so we can apply this to the right side of Equation~\eqref{eq:p_bound_1}:
\begin{equation}
    1 - \left(1 - \frac{\delta}{2}\right)^\frac{1}{R} \ge \frac{\delta}{2R} 
\end{equation}
This gives us our new, more restrictive constraint for $1 - p$:
\begin{equation}
    1 - p \le \frac{\delta}{2R} \label{eq:p_final}
\end{equation}
\textbf{Analysis of $Q$:} From Theorem~\ref{thm:distance_sensitive_Bloom}, we have an upper bound on the probability that a single point is falsely reported, $\mathrm{Pr}[\mathrm{Report}\, x | x \not \in K]$. Using the union bound, we get that $Q$, the probability that \textit{any} of the $N$ points is falsely reported, is less than or equal to $N$ times the probability that a single point is falsely reported:
\begin{equation}
    Q \le N * \mathrm{Pr}[\mathrm{Report}\, x | x \not \in K]^N \label{eq:qu_bound} 
\end{equation}
Thus, we can achieve the requirement from Equation~\eqref{eq:half_bound} that $Q \le \frac{\delta}{2}$ by requiring that
\begin{equation}
    N * \mathrm{Pr}[\mathrm{Report}\, x | x \not \in K]^N < \frac{\delta}{2}
\end{equation}
If we now substitute in our expression for $\mathrm{Pr}[\mathrm{Report}\, x | x \not \in K]$ with $|K| = 1$ and $B = 2\sqrt{N}$ from Theorem~\ref{thm:distance_sensitive_Bloom} and (extensively) simplify, we have a constraint for $q$.
\begin{align}
\begin{split}
    q < \frac{2eN^\frac{1}{2}}{e^2(2N^\frac{1}{2} - 1)}\Bigg[&\frac{N - 1}{N} \left(\frac{\delta}{2N}\right)^\frac{1}{R}\\ - &\frac{p[(2e\frac{N - 1}{N} - 2)N^\frac{1}{2} + 1]}{2eN^\frac{1}{2}}\Bigg]
    \label{eq:complicated_constraint_q}
\end{split}
\end{align}
Like we did above for $1-p$, we can now tighten this constraint for $q$ to obtain a simpler expression. The simplified constraint leads to a smaller range of values for $q$, but these values still satisfy the original constraint and guarantee a total error rate of $\delta$. We decrease the constraint by replacing the factor of $-p$ with $-1$, by replacing the factor of $-\frac{N - 1}{N}$ with $-1$, and by replacing the factor of $\frac{e(2N^\frac{1}{2})}{e(2N^\frac{1}{2} - 1)}$ with $1$. We end up with the following (tighter) constraint for $q$:
\begin{equation}
    q < \frac{N - 1}{Ne}\left(\frac{\delta}{2N}\right)^\frac{1}{R} - \frac{(2e - 2)N^\frac{1}{2} + 1}{2e^2N^\frac{1}{2}} \label{eq:q_bound_1}
\end{equation}
Breaking up the $\frac{\delta}{2N}$ term and simplifying, we get
\begin{equation}
    q < \left(\frac{N - 1}{N}\right) \frac{2^{1 - \frac{1}{R}}N^{\frac{1}{2} - \frac{1}{R}}\delta^\frac{1}{R}}{2eN^\frac{1}{2}} - \frac{(2e - 2)N^\frac{1}{2} + 1}{2e^2N^\frac{1}{2}}
\end{equation}
Note that as $R$ increases, the right hand side of the constraint for $q$ also increases. For some small values of $R$, the right hand side is actually negative. A negative expression means that we have shrunk the range of allowable $q$ values so much that our simplified constraint is no longer meaningful. This does not mean that it is impossible to find $q$ to satisfy the original constraint, it simply means that our simplifications were too aggressive. 

We now make two key assumptions that allow us to show that the right hand side of Equation~\eqref{eq:q_bound_1} is always positive and well defined:
\begin{equation}
    N \ge 150 \qquad R \ge 10\log N > 50
\end{equation}
The analysis is possible without these assumptions, but must be done with the complicated expression in Equation~\ref{eq:complicated_constraint_q} rather than the simple version.

We continue to tighten the constraint by replacing some values of $N$ and $R$ with their smallest possible values (i.e. $N = 150$ and $R = 50$), in cases where making such a replacement will only make the right hand side of the constraint smaller\footnote{We make the substitution whenever the replaced term monotonically increases with increasing $N$ and $R$}:
\begin{equation}
q < \frac{149}{150}  \frac{2^{\frac{49}{50}}N^{\frac{1}{2} - \frac{1}{10\log N}} \delta^\frac{1}{R}}{2eN^\frac{1}{2}} - \frac{(2e - 2)N^\frac{1}{2} + 1}{2e^2N^\frac{1}{2}}
\end{equation}
We also simplify the term  $N^{\frac{1}{2} - \frac{1}{10\log N}}$:
\begin{align*}
    N^{\frac{1}{2} - \frac{1}{10\log N}}
    &= N^\frac{1}{2}  N^\frac{-1}{10\log N}\\
    &= N^\frac{1}{2}  (N^\frac{1}{\log N})^\frac{-1}{10}\\
    &= N^\frac{1}{2}e^\frac{-1}{10} &\text{since $x^\frac{1}{\log x} = e$}
\end{align*}
Plugging this value back into the constraint, we have
\begin{equation}
q <  \frac{\frac{149}{150}2^{\frac{49}{50}}N^{\frac{1}{2}}e^\frac{-1}{10} \delta^\frac{1}{R}}{2eN^\frac{1}{2}} - \frac{(2e - 2)N^\frac{1}{2} + 1}{2e^2N^\frac{1}{2}}
\end{equation}
To obtain our final constraint for $q$, we first combine fractions by multiplying the top and bottom of the left fraction by $e$:
\begin{equation}
q < \frac{\frac{149}{150}2^{\frac{49}{50}}N^{\frac{1}{2}}e^\frac{9}{10} \delta^\frac{1}{R} - (2e - 2)N^\frac{1}{2} - 1}{2e^2N^\frac{1}{2}}
\end{equation}
We then explicitly calculate the constants in the numerators, and slightly tighten the constraint by rounding the constants up/down appropriately. Since we are tightening the constraint, we can also replace the "less than" with a "less than or equal to." We finally get a simple constraint for $q$, such that any $q$ that satisfies the below inequality will solve the $\delta$ nearest neighbor problem:
\begin{equation}
    q \le \frac{4.80N^{\frac{1}{2}}\delta^\frac{1}{R} - 3.44N^\frac{1}{2} - 1}{2e^2N^\frac{1}{2}} \label{eq:q_final}
\end{equation}
Notice that this bound for $q$ is positive when $R$ is sufficiently large, since $\delta^\frac{1}{R}$ approaches $1$ and the numerator approaches the positive value $1.36N^\frac{1}{2} - 1$. 

\textbf{Solving for $p$, $q$, and $R$:} 
We now fix $p$ and $q$ to be the largest (and thus, least expensive) values that obey their respective simplified constraints. Using the edge of the $1 - p$ constraint range from Equation~\eqref{eq:p_final} and the edge of the $q$ constraint range from Equation~\eqref{eq:q_final}, we set
\begin{align}
    p &= 1 - \frac{\delta}{2R} \label{eq:p_explicit}\\
    q &= \frac{4.80N^{\frac{1}{2}}\delta^\frac{1}{R} - 3.44N^\frac{1}{2} - 1}{2e^2N^\frac{1}{2}} \label{eq:q_explicit_1}
\end{align}
Notice that -- although our analysis is performed under the assumption that $R > 10\log N$ -- we may still choose a value for the free parameter $R$. Our strategy is to select a value of $q$ which satisfies the constraint, and then solve for an $R$ which guarantees this value of $q$. To simplify our later analysis, we will use 
\begin{align}
    q = \frac{1}{\sqrt{N}} \label{eq:q_explicit_2}
\end{align}
We can now plug in Equation~\eqref{eq:q_explicit_2} into Equation~\eqref{eq:q_explicit_1} and solve for $R$:
\begin{equation}
    R = \frac{\log(\frac{1}{\delta})}{(\log (4.80N^\frac{1}{2}) - \log (2e^2 + 3.44N^\frac{1}{2}))} \label{eq:r_explicit}
\end{equation}
Note the right side is a valid fraction less than $1$ because $N >= 150$. In fact, the denominator is fixed within a small range: the largest it can be is when $N = 150$, when it is about $0.97$, and the smallest it can be is about $0.72$, when $N\to\infty$. Thus the denominator is $O(1)$. 
    
When $\delta$ is small enough that $R > 10\log N$, then Equation~\eqref{eq:r_explicit} yields a value of $R$ which satisfies our simplifying assumption and attains the correct value of $q$. 

We have now found explicit values of $p$, $q$ and $R$, which constrain all of the free parameters of our data structure. These values are given by Equation~\eqref{eq:p_explicit}, Equation~\eqref{eq:q_explicit_2}, and Equation~\eqref{eq:r_explicit}. We have shown that these values attain a sufficiently low false positive rate $Q$ and high true positive rate $P$ to solve the nearest neighbor problem, proving the theorem. 
\end{proof}

We have specific parameter settings from Lemma~\ref{lem:cellwise_rate_bounds} that solve the nearest neighbor problem, but it remains to bound the query time of the resulting data structure. We obtain our main theorem by using these values with the query time expresssion from Lemma~\ref{lem:group_distance_sensitive_query_time}.
\begin{thm} \textbf{(Main Theorem)}
\label{thm:final_query_time}
Under the conditions in Lemma \ref{lem:cellwise_rate_bounds}, we solve the randomized nearest neighbor problem for $\gamma$-stable queries in time $t_{query}$ with probability $1 - \delta$.
\begin{align}
t_{query} = O\left(N^{\frac{1}{2} + \gamma}\log^4(N)\log^3\left(\frac{1}{\delta}\right)\right)
\end{align}
\end{thm}
\begin{proof}
We made many simplifications to the constraints for $p$ and $q$ in Lemma~\ref{lem:cellwise_rate_bounds}. These simplifications require us to solve a \textit{harder} version of the problem than necessary. For example, there is some maximum value of $\delta$ in Lemma~\ref{lem:cellwise_rate_bounds} that suffices to make $R \ge 10\log N$. Call this $\delta'$, such that when $\delta = \delta'$ the value of $R$ from Lemma~\ref{lem:cellwise_rate_bounds} is $10\log N$. To solve the nearest neighbor problem for an arbitrary $\delta$, we split our analysis into two cases, $\delta \ge \delta'$ and $\delta < \delta'$, and solve for the query time under each one.

\textbf{Case 1, $\delta \ge \delta'$:} If $\delta \ge \delta'$, we use the values from Lemma~\ref{lem:cellwise_rate_bounds} with $\delta = \delta'$. This will return an array of tests that solves the $\delta'$ nearest neighbor problem. Because we substantially simplified the constraints in Lemma~\ref{lem:cellwise_rate_bounds}, $\delta' < \delta$ and thus we solve a harder version of the problem than necessary. When $\delta = \delta'$, we have the following from Lemma~\ref{lem:cellwise_rate_bounds}:
\begin{align*}
    R &= 10\log N\\
    q &= N^{-\frac{1}{2}}\\
    1 - p &= \frac{\delta'}{2R}\\
    |K| &= 1
\end{align*}
We now plug these values into our query time result from Lemma~\ref{lem:group_distance_sensitive_query_time} and simplify. Note that Lemma~\ref{lem:group_distance_sensitive_query_time} has its \textit{own} failure probability $\delta_3$. Note that to have $R = 10\log N$, $\frac{1}{\delta'}$ is slightly smaller than $N^{10}$, so $\log(\frac{1}{\delta'}) = O(\log N)$. This leaves us with
\begin{equation}
    t_{query} = O(N^{\frac{1}{2}+\gamma}\log^4(N)\log (1 / \delta_3) ) \label{eq:case_1_result}
\end{equation}
\textbf{Case $2$, $\delta < \delta'$:} If $\delta < \delta'$, then our simplifying changes to the constraints no longer force us to solve a harder problem than necessary. In this case, we use the values from Lemma~\ref{lem:cellwise_rate_bounds} using $\delta$. As before, we obtain our query time result from Lemma~\ref{lem:group_distance_sensitive_query_time}:
\begin{align}
\begin{split}
     t_{query} = O&\bigg(\log\left(\frac{1}{\delta}\right)N^{\frac{1}{2} + \gamma}\log^2(N)\\&\max \left(\log(N^\frac{1}{2}), \log\left(\frac{R}{\delta}\right)\right)\log (1 / \delta_3) \bigg) \label{eq:i_give_up_on_naming_these}
\end{split}
\end{align}
Note we used the fact that $R = O(\log(\frac{1}{\delta}))$, since as we noted in the proof of Lemma~\ref{lem:cellwise_rate_bounds} the denominator in the equation for $R$ is $O(1)$. We can further simplify Equation~\eqref{eq:i_give_up_on_naming_these} by rewriting $\log\left(\frac{R}{\delta}\right)$ as $\log(R) + \log\left(\frac{1}{\delta}\right)$, and recognizing that $\log(R)$ is dominated by $\log\left(\frac{1}{\delta}\right)$. Furthermore, since both terms in the $\max$ operation are greater than $1$, we note that the maximum is asymptotically smaller than the product of the two terms. After simplifying, we have that
\begin{equation}
    E[t_{query}] = O\left(\log^2\left(\frac{1}{\delta}\right)N^{\frac{1}{2} + \gamma}\log^2(N)\right)
    \label{eq:case_2_result}
\end{equation}
\textbf{Combining Results:}
In total the runtime for an arbitrary $\delta$ is the maximum of case $1$ (Equation~\eqref{eq:case_1_result}) and case $2$ (Equation~\eqref{eq:case_2_result}):
\begin{equation}
    t_{query} = O\left(N^{\frac{1}{2} + \gamma}\log^4(N)\log^2\left(\frac{1}{\delta}\right)\right)
\end{equation}
Here, we have absorbed $\delta_3$ into $\delta$, which adds only a constant multiplier to the expression.
\end{proof}

\section{Experiments}
In this section, we provide additional details about our experiments. We also show a full table of index characteristics for each dataset in our evaluation.

\begin{figure*}[t]
\begin{center}
\centerline{\includegraphics[width=\textwidth]{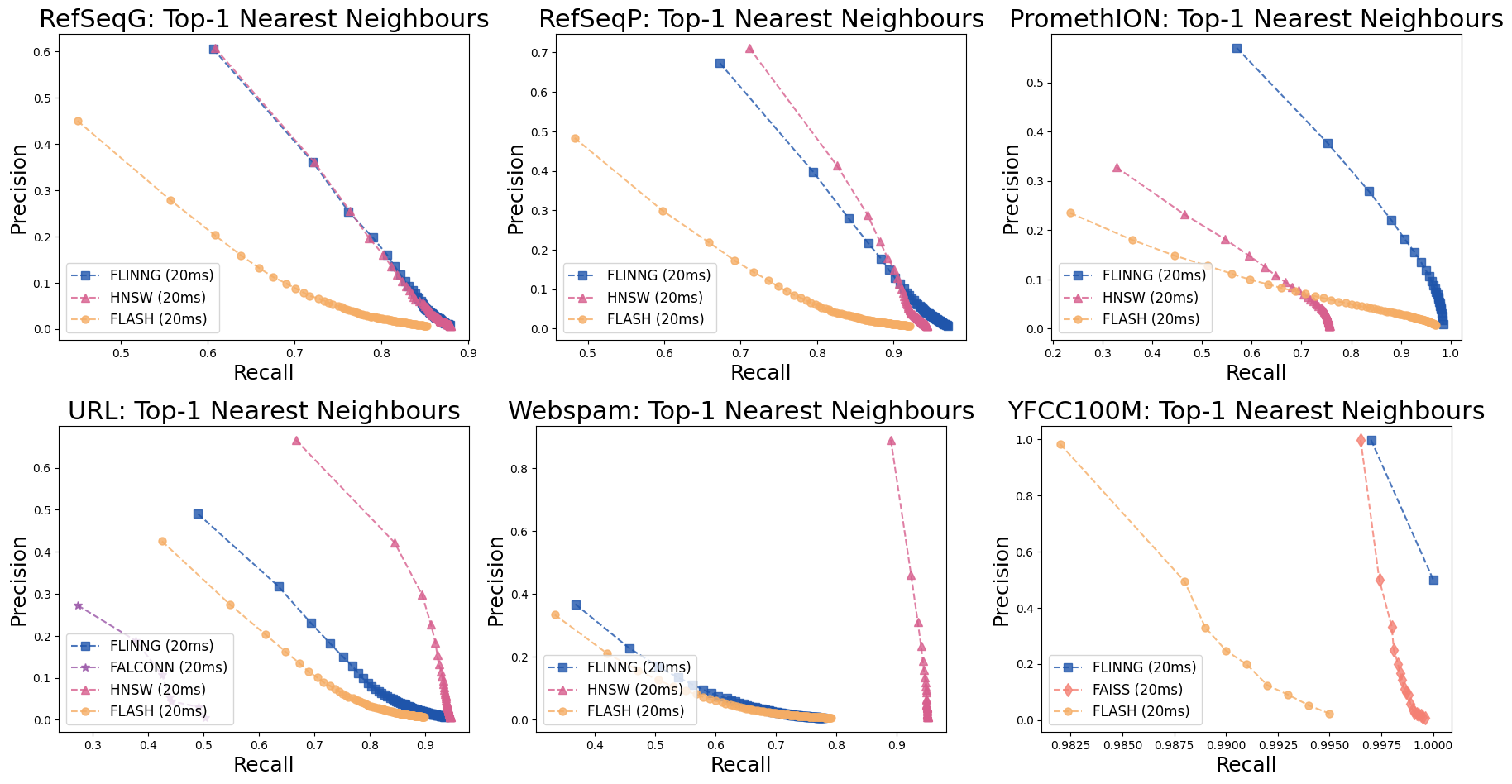}}
\caption{Precision recall graphs for the top $1$ nearest neighbours for each dataset that we tested on.}
\label{fig:precall}
\end{center}
\vspace{-1cm}
\end{figure*}

\begin{figure*}[t]
\vspace{-0.4cm}
\begin{center}
\centerline{\includegraphics[width=\textwidth]{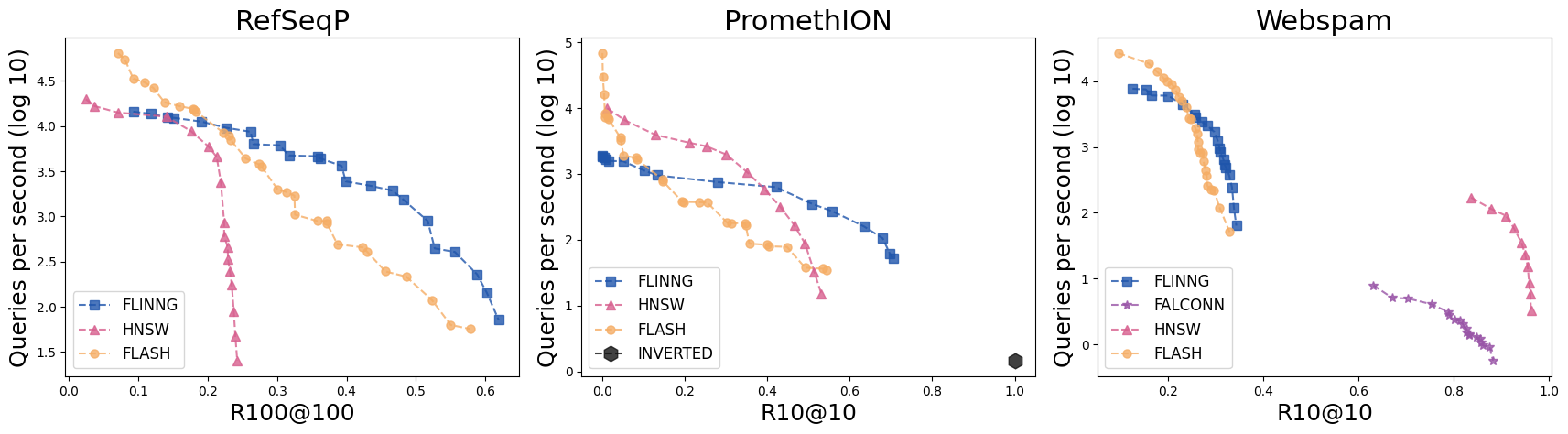}}
\caption{Plots for top-$10$ and top-$100$ nearest neighbor search on selected datasets. FLINNG performs best when the top neighbors are very similar to the query, as predicted by the theory.}
\label{fig:top100}
\end{center}
\vspace{-1cm}
\end{figure*}

\begin{table*}
\caption{Algorithm index sizes in gigabytes, determined by the minimum index size an algorithm achieved over all hyperparameters, subject to the restriction in the restriction column.}
\label{sizes}
\label{times}
\begin{center}
\small\addtolength{\tabcolsep}{-1pt}
\begin{sc}
\begin{tabular}{lcccccccr}
\toprule
Dataset & Restriction & FLINNG & FAISS & FALCONN & HNSW & FLASH & INVERTED & GROUPS\\
\midrule
RefSeqG &  $R1@100 > 0.8$ & $0.090$ & -& -& $0.031$& $4.29$& - & -\\
RefSeqP & $R1@100 > 0.8$ & $0.090$ & - & -& $0.031$& $4.29$& - & -\\
PromethION & $R1@100 > 0.8$ & $0.074$ & - & -& $0.64$ &$8.59$ & $4.22$ & -  \\
URL & $R1@100 > 0.5$ & $0.048$ & - & -& $3.03$& $4.29$& $2.21$ & -\\
Webspam & $R1@100 > 0.5$ & $0.0068$ & - & $2.25$ & $8.04$ & $4.29$ & - & $0.14$\\
YFCC100M & $R1@1 > 0.95$ & $3.5$ & $3.7$ & -& -& $4.29$ & - & -\\
\bottomrule
\end{tabular}
\end{sc}
\end{center}
\vskip -0.1in
\end{table*}

\subsection{System Details:} We performed all experiments using $1.48$ TB of RAM. For YFCC100M, we used 88 Intel Xeon E5-2699A v4 processors, each of which has a $56$ MB L3 cache. For YFCC100M, we used 96 Intel Xeon Gold 5220R processors with a 36 MB cache. 

\subsection{Baseline Failures}
Some of the baseline methods did not run on some of the datasets. We attempted to construct HNSW and FALCONN indices on YFCC100M, but memory limitations meant that the index would not fit within our 1.48 TB of RAM. We modified the HNSW library to work with nonstandard short floating point vectors, but the resulting index took more than 5 days to construct. We also tried to build FAISS indices on the genomics and web datasets, but were unable to fit the quantized data in memory because these problems are ultra high-dimensional.

\subsection{Hyperparameters}
For algorithms with fast indexing times like FLINNG, FLASH, FALCONN, and the grouping algorithm from~\cite{shi2014group}, we tried hundreds of hyperparameter settings and selected the best indices. For algorithms such as HNSW and FAISS, which can take hours or days to train, we built indices using suggested parameters and tuned the query-time arguments aggressively. 

FLINNG requires four hyperparameters: $R$, $B$, $m$, and $L$. We use $R = \{2,3,4\}$ and $B = 2^a$ for $a \in [11,15]$. To have 16-bit cell IDs, we constrain $BR < 2^{16}$. For YFCC100M, we use $R = 2$ and $a \in [12,19]$. We set the number of LSH functions $m$ to $2^a$ for $a \in [2,11]$. We used $L = 18$ for Webspam, $L = 12$ for YFCC100M and $L = 17$ for the other datasets.


FLASH requires the following hyperparameters: $m$ (the number of hash tables), $L$ (the number of hash functions in each table), and $s$ (the size of each reservoir for reservoir sampling). We used the recommendations from the paper. However, we found that much larger values of $m$ and $s$ were needed than in the original paper to obtain high recall on some of our tasks. We used the same values of $m$ as the authors of FLASH: $m = 2^a$ for $a \in [2, 11]$ for all datasets. For the non YFCC100M datasets, we let $s = 2^a$ for $a \in [2, 11]$. For YFCC100M, we let $s = 2^a$ for $a \in [3, 12]$. As with FLINNG, we used $18$ hash bits for webspam, $12$ hash bits for YFCC100M, and $17$ hash bits for every other dataset.

Our implementation of the grouping algorithm from \cite{shi2014group}, which we denote GROUPS, requires $M$, the number of groups, $t$, the number of back propagation steps, $N_L$, the number of groups each point is in, and $R$, the total number of points to checks across all $t$ back propagation steps (see the original paper for more details on each parameter). None of our datasets in dense format fit in memory, which GROUPS requires, but we were able to project the URL and Webspam datasets into $400$ dimensions using $400$ random projections to get a meaningful benchmark against our algorithm (we cannot apply random projections to the other datasets so we were not able to run GROUPS on the other datasets). For Webspam, we tried all combinations of $M = 20000, 40000$, $t = 1, 2, 4$, $N_L = 2, 4$, and $R / t = 10000, 80000$. For URL, we tried all combinations of $M = 20000, 80000, 320000$, $t = 1, 2, 4$, $N_L = 2, 4$, and $R / t = 10000, 20000, 80000$.

FALCONN could only run on the URL and Webspam datasets because the package does not natively support Jaccard similarity for the genome datasets and has out-of-memory issues for YFCC100M. FALCONN requires three hyperparameters: $m$ (the number of hash tables), $n_p$ (the number of probes for multi-probe LSH), and $L$ (the number of hash functions for each table). We let $n_p = 2^a$ for $a \in [1, 9]$, $m = 2^a$ for $a \in [1, 4]$, and used $22, 20, 18$ hash bits for URL and $20, 18, 16$ hash bits for Webspam. 

Due to the high dimensionality of the other datasets, FAISS was only feasible for YFCC100M. For high-dimensional sparse data such as Webspam or URL, quantization actually \textit{increases} the memory of the index. We used an inverted file index with product quantization. This index requires two construction parameters: $m$ (the number of $k$-means centroids) and $s$ (the number of bits for product quantization). There is one query-time parameter $n_p$ (the number of clusters probed for each query). FAISS also supports the use of an HNSW graph to identify the best clusters, so we use both flat (i.e. brute force) and HNSW pre-indexing. We trained $3$ different indices, all with $s = 32$ bit product quantization: $m = 4$k centroids with flat (brute force) indexing, $m = 262$k centroids with flat indexing, and $65$k with HNSW indexing. We used a subset of one million points to train the $k$-means centroids. We used $n_p = 2^a$ for $a \in [1, 9]$. 

HNSW requires two construction hyperparameters: $ef_c$ and $M$. The parameter $M$ is the maximum number of edges for each node in the graph, while $ef_c$ may be thought of as a parameter that controls the quality of the near neighbor graph (larger is better). We used parameter $M = 32$ and $ef_c = 100$ for all trials. HNSW has one query-time parameter $ef_s$, which controls the recall-time tradeoff. We let $ef_s = 2^a$ for $a \in [4, 14]$. For our genomics datasets (PromethION and RefSeq), we used a pregenerated and fixed number of minhashes to allow HNSW to perform fast search in the Jaccard metric space. We used $2^a$ for $a \in [1, 11]$ number of hashes. We modified the HNSW code to support the approximate Jaccard metric by implementing a distance functin that counts the collisions among these hashes for two sequences. We tried to build an index for YFCC100M, but the graph construction algorithm did not finish,even after four days of construction time. 

Finally, we used inverted indices to compute the ground truth results for the URL and PromethION datasets. We show the query time for this structure as a baseline. The other genomic datasets (RefSeqG and RefSeqP) were too high-dimensional for an inverted index lookup to be practical, and Webspam and YFCC100M had too many nonzeros.


\subsection{Supplementary Plots}

Figure~\ref{fig:precall} shows the precision recall plots for all datasets considered in our evaluation. Figure~\ref{fig:top100} shows the latency-recall relationship for top-10 and top-100 near neighbor search on selected datasets.


\subsection{Index Characteristics}
Table~\ref{sizes} shows the index size for all of the indices considered in our evaluation.

\end{document}